\documentclass[journal,12pt,onecolumn]{IEEEtran}

\ifCLASSINFOpdf
\else
\fi
\usepackage[explicit]{titlesec}
\usepackage{graphicx}
\usepackage{epstopdf}
\usepackage{subfigure} 
\usepackage{algorithmic}
\usepackage{setspace}
\usepackage{amsmath}
\usepackage{amsthm}
\usepackage{cite}

\newtheorem{case}{Case}
\usepackage{amssymb}
\newtheorem{remark}{Remark}

\newcommand{\RNum}[1]{\uppercase\expandafter{\romannumeral #1\relax}}
\newtheorem{lemma}{Lemma}
\newtheorem{definition}{Definition}
\newtheorem{corollary}{Corollary}
\newtheorem{proposition}{Proposition}
\usepackage{stfloats}
\newtheorem{problem}{Problem}
\usepackage{color}
\usepackage{tikz,xcolor,hyperref}
\definecolor{lime}{HTML}{A6CE39}%
\newcommand{\tabincell}[2]{\begin{tabular}{@{}#1@{}}#2\end{tabular}}
\titlespacing{\section}{0pt}{0ex plus .0ex minus .0ex}{.3ex plus .0ex}
\titlespacing{\subsection}{0pt}{0ex plus .0ex minus .0ex}{.3ex plus .0ex}

\makeatletter
\DeclareRobustCommand{\orcidicon}{%
	\begin{tikzpicture}
	\draw[lime, fill=lime] (0,0) 
	circle [radius=0.16] 
	node[white] {{\fontfamily{qag}\selectfont \tiny ID}};    \draw[white, fill=white] (-0.0625,0.095) 
	circle [radius=0.007];    \end{tikzpicture}
	\hspace{-2mm}}
\foreach \x in {A, ..., Z}{%
	\expandafter\xdef\csname orcid\x\endcsname{\noexpand\href{https://orcid.org/\csname orcidauthor\x\endcsname}{\noexpand\orcidicon}}
}

\newcommand*\bigcdot{\mathpalette\bigcdot@{.5}}
\newcommand*\bigcdot@[2]{\mathbin{\vcenter{\hbox{\scalebox{#2}{$\m@th#1\bullet$}}}}}
\makeatother

\usepackage[linesnumbered,ruled,vlined]{algorithm2e}

\begin{document}
	%
	\title{Age of Information with Hybrid-ARQ: A Unified Explicit Result}
	%
	%
	%
	
	\author{Aimin Li\orcidA, 
		\emph{Graduate Student Member, IEEE,}
		Shaohua Wu\orcidB,
		\emph{Member, IEEE,}
		Jian Jiao\orcidD, 
		\emph{Member, IEEE,}
		Ning Zhang\orcidE, 
		\emph{Senior Member, IEEE,}
		and Qinyu Zhang\orcidF
		\emph{Senior Member, IEEE.}
		\thanks{
			This work was supported in part by the National Natural Science Foundation of China under Grant nos. 61871147, 61831008, 91638204, and in part by the Shenzhen Municipal Science and Technology Plan under Grant nos. JCYJ20170811160142808, JCYJ20170811154309920. \emph{(Corresponding author: Shaohua Wu.)}
			
			A. Li, S. Wu, J. Jiao and Q. Zhang are with the Department of Electronic Engineering, Harbin Institute of Technology (Shenzhen), Guangdong, China. N. Zhang is with the Department of Electrical and Computer Engineering, University of Windsor, Windsor, ON, N9B 3P4, Canada (e-mail: liaimin@stu.hit.edu.cn; hitwush@hit.edu.cn; jiaojian@hit.edu.cn; ning.zhang@uwindsor.ca; zqy@hit.edu.cn).
		}
	}

	\maketitle
	\begin{abstract}
		Delivering timely status updates in a timeliness-critical communication system is of paramount importance to assist accurate and efficient decision making. Therefore, the topic of analyzing \emph{Age of Information (AoI)} has aroused new research interest. This paper contributes to new results in this area by systematically analyzing the AoI of two types of Hybrid Automatic Repeat reQuest (HARQ) techniques that have been newly standardized in the Release-16 5G New Radio (NR) specifications, namely reactive HARQ and proactive HARQ. Under a code-based status update system with non-trivial coding delay, transmission delay, propagation delay, decoding delay, and feedback delay, we derive unified closed-form average AoI and average Peak AoI expressions for reactive HARQ and proactive HARQ, respectively. Based on the obtained explicit expressions, we formulate an AoI minimization problem to investigate the age-optimal codeblock assignment strategy in the finite block-length (FBL) regime. Through case studies and analytical results, we provide comparative insights between reactive HARQ and proactive HARQ from a perspective of \emph{freshness} of information. The numerical results and optimization solutions show that proactive HARQ draws its strength from both age performance and system robustness, thus enabling the potential to provide new system advancement of a freshness-critical status update system.
	\end{abstract}
	
	
	\IEEEpeerreviewmaketitle

	\section{Introduction}
	\subsection{Background}
	
	In recent ten years since Kaul \emph{et al.} proposed a framework to quantify the timeliness of information in 2012 \cite{Zerowait2}, one of the most popular ideas in timely update system design has been how to keep information as fresh as possible and ensure timely information delivery. For timely update systems such as vehicle networks where the vehicle's velocity and location are disseminated to ensure safe transportation \cite{vehicle}, environmental sensor networks where the updates of a time-varying phenomenon are collected for large-scale monitoring \cite{sensor}, and wireless communication networks where adaptive scheduling algorithms are adopted based on the time-varying channel state information \cite{channel}, achieving timely delivery can freshen the monitor’s awareness of the sources and thus assist correct and efficient decision making. 
	
	This has aroused new interest in analyzing \emph{Age of Information (AoI)} performance metrics. AoI has been broadly used to capture the \emph{freshness} of a monitor's knowledge of an entity or process. Different from conventional performance metrics such as delay and throughput, AoI comprehensively measures the effects of update rate, latency, and system utilization. Initial works on this issue were mainly based on queue analysis, which originated from single-source single-server queues \cite{packetmagnament1,packetmagnament2,Zerowait,Zerowait2}, and subsequently developed to multiple-sources single-server queues \cite{ISITMSDS,ISITMSDS2,ISITMSDS3,ISITMSDS4} and wireless queuing networks \cite{ISITNW1,ISITNW2,ISITNW3,ISITNW4,ISITNW5,ISITNW6}. These works are based on an ideal assumption that the status update is transmitted through a perfect channel without packet errors and losses. In practice, however, packet errors and losses are inevitable due to ubiquitous noises, signal interference, and channel fading. As the incorrectly decoded message does not bring about \emph{fresh} awareness, the packet errors and losses will result in \emph{staleness} of information, leading to uncontrollable residual errors, system instability, and wrong decisions. Therefore, it is imperative to analyze the AoI over unreliable channels.
	
	\subsection{Related Works}
	Some recent works have noticed the above limitation and have extended the AoI analyses to the physical (PHY) layer. One pioneering work concerning this issue was accomplished by Chen, \emph{et al.} in 2016 \cite{Channelcoding3}, in which the update is delivered over an erasure channel and the Peak Age of Information (PAoI) is studied. This work has aroused extensive research interest in understanding the effect of system reliability on AoI. From then on, including but not limited to the follow-up works that also analyzed AoI over the erasure channel \cite{Channelcoding4,Channelcoding5,Blockoptim}, various transmission protocols, ranging from conventional protocols like non-ARQ, classical ARQ and truncated ARQ protocols, to state-of-the-art protocols such as HARQ with Chase Combing (HARQ-CC) and HARQ with Incremental Redundancy (HARQ-IR) protocols, have been investigated under different types of noisy channels \cite{Age_LDPC,Channelcoding2,Channelcoding1,FBL_AoI,ShortAoI_2}.
	
	We notice that the above AoI analyses focus on the transmission delay, and neglect other types of system delay such as coding delay, propagation delay, decoding delay and feedback delay. An exception work is \cite{Channelcoding_delay1}, which considers non-trivial propagation delay and studies the AoI of HARQ-IR with a fixed number of retransmitted packets $m=2$ under Satellite-IoT Systems, but also assumes negligible coding delay and decoding delay. Nevertheless, in practical communication systems, especially the short-packet communication, the coding delay and decoding delay are also nontrivial compared to the transmission delay, resulting in the \emph{staleness} of information by nature. Thus, we focus on a more realistic (or general) scenario where different types of delay elements naturally exist and the number of retransmitted packets is not fixed to $m=2$. In this regard, we would like to provide a basic framework to comprehensively study the trade-off among coding complexity, decoding complexity, code length, number of retransmitted packets and error probability from the AoI perspective.

	Up to this point, we have only introduced AoI research based on conventional \emph{reactive} HARQ (also known as stop-and-wait HARQ), which allows for retransmissions only upon the reception of a Negative ACKnowledgment (NACK). As such, the retransmission process is not truly automatic. In the Release-16 5G NR specifications by the $3^{rd}$ Generation Partnership Project (3GPP), a new HARQ protocol named \emph{proactive} HARQ is designated for the Up-Link Grant-Free communication to enable the potential for meeting the stringent requirements for URLLC \cite{Proactive2}. Some recent works have shown the superiority of \emph{proactive} HARQ in terms of latency and throughput compared to \emph{reactive} HARQ \cite{GFHARQ1,GFHARQ2,GFHARQ3}. Inspiringly, these available studies also witness the potential for \emph{proactive} HARQ to be applied in the freshness-critical status update system. To this end, we would like to theoretically analyze the AoI performance of \emph{reactive} HARQ and \emph{proactive} HARQ to investigate whether \emph{proactive} HARQ will facilitate timeliness of information in the freshness-critical status update system.

	\begin{table*}[t]
		\normalsize
		\centering
		\renewcommand\arraystretch{1.5}
		\caption{Construsting The Novelty Of Our Work To The Literature}
		\setlength{\tabcolsep}{0.5mm}{
			\begin{tabular}{c|c|c|c|c|c|c}
				\hline
				Contributions & {\bfseries This Work} & \cite{Zerowait2,vehicle,sensor,channel,packetmagnament1,packetmagnament2,Zerowait,ISITMSDS,ISITMSDS2,ISITMSDS3,ISITMSDS4,ISITNW1,ISITNW2,ISITNW3,ISITNW4,ISITNW5,ISITNW6} & \cite{Channelcoding3,Channelcoding4,Channelcoding5,Age_LDPC,Channelcoding2} & \cite{FBL_AoI,ShortAoI_2} & \cite{Blockoptim,Channelcoding1} & \cite{Channelcoding_delay1} \\ \hline \hline
				Age of Information (AoI) & $\checkmark$ & $\checkmark$ & $\checkmark$ & $\checkmark$ & $\checkmark$ & $\checkmark$ \\ \hline	
				Finite Block-Length Regime & $\checkmark$ & ~ & ~ & $\checkmark$ & ~ & ~ \\ \hline
				Reactive Hybrid-ARQ & $\checkmark$ & ~ & $\checkmark$ & $\checkmark$ & $\checkmark$ & $\checkmark$ \\ \hline
				Flexible Number of Retransmissions& $\checkmark$ & ~ & ~ & ~ & $\checkmark$ & ~ \\ \hline
				\tabincell{c}{Effect of Delay Elements \\Other Than Transmission Delay} & $\checkmark$ & ~ & ~ & ~ & ~ & $\checkmark$ \\ \hline
				Proactive Hybrid-ARQ & $\checkmark$ & ~ & ~ & ~ & ~ & ~ \\ \hline
		\end{tabular}}
	\end{table*}
	\subsection{Contributions}
	The research on the HARQ-based timely status update system is still in the ascendant, and some open issues remain to be addressed. First, there have been a lot of works providing explicit average age results under different types of protocols and systems. Examples include the average AoI expressions under fixed-length non-ARQ protocols, truncated-ARQ, classical ARQ, and the explicit results of some advanced ARQ-based techniques like HARQ-CC and HARQ-IR. However, there has not been a unified expression that can unify the aforementioned expressions in a single closed-form formula. By providing such a unified result, the comparative insights and the intrinsic relationships among different protocols will be further investigated. Second, the existing literature only considers certain types of delay in the status update system and assumes others to be negligible. However, as the delay exists by nature and plays as a critical part in affecting the \emph{freshness} of information, to comprehensively consider the coding delay (or processing delay), propagation delay, transmission delay, decoding delay and feedback delay in the status update system and provide a unified closed-form result will provide a systematic understanding in analyzing the age of a realistic freshness-critical status update system. Third, the age performance has been extensively studied over erasure channels. However, little research considers the short-packet AoI analysis over the AWGN channel. Finally, the majority of existing works mainly focus on AoI analysis of conventional \emph{reactive} HARQ. Some recent works analyzing the performance of \emph{proactive} HARQ are based on some conventional performance metrics, such as throughput and latency. Thus, to analyze the AoI of \emph{proactive} HARQ will fill this research gap and may further facilitate new system advancement of a status update system.
	
	Motivated by the above, this work achieves several key contributions and we summarize them as follows:
	\begin{itemize}
		\item We derive unified closed-form average AoI and average Peak AoI expressions for \emph{reactive} HARQ, wherein: $i$) different kinds of delay elements (i.e., coding delay, transmission delay, propagation delay, decoding delay, and feedback delay) are comprehensively considered; $ii$) the number of repeated packets is not fixed as $m=2$, but is relaxed to a variable value; $iii$) different types of protocols are unified to a single expression. 
		\item We investigate the AoI explicit expressions and comparative insights for \emph{proactive} HARQ, which is the first work analyzing \emph{proactive} HARQ from the AoI perspective. Theoretical and numerical comparisons are given to show the superiority of \emph{proactive} HARQ in enabling timely information delivery.
		\item We also try to further optimize the AoI for both \emph{reactive} HARQ and \emph{proactive} HARQ. By formulating an AoI minimization problem in the FBL regime, we solve out the age-optimal block assignment strategy for \emph{reactive} HARQ and \emph{proactive} HARQ, respectively. The results show that the optimal strategy for \emph{proactive} HARQ turns out to be the finest grained symbol-by-symbol transmission, while that for \emph{reactive} HARQ is highly dependent on the propagation delay and SNR.
	\end{itemize}
	
	\subsection{Organization}
	The rest of this paper is organized as follows. In Section II, we briefly introduce the considered system model. The generalized closed-form expressions of average AoI and average Peak AoI for reactive HARQ and proactive HARQ are provided in Section III, where the effect of different delay elements is added into the analysis. In Section IV,  we design an optimization problem to reduce the average AoI of reactive HARQ and proactive HARQ, respectively. Numerical results and discussions are given in Section V, followed by conclusions in Section VI. 
	
	\section{System Model}
	We consider an end-to-end (E2E) code-based timely status update system. The update generator (source) is monitoring a time-varying phenomenon $\mathcal{F}\left(t\right)\in\left\{0,\cdots,2^k-1\right\}$, where the time $t$ is divided into some time slots in units of channel use such that $t \in \mathbb{N}$\footnote{Here we consider the symbol-level AoI analysis. Some recent works focusing on PHY-layer AoI analysis also discretize the time into time slots to analyze symbol-level AoI \cite{Channelcoding4,Blockoptim,Channelcoding2,Age_LDPC}.}. We assume that the monitored phenomenon is modeled as a sequence of independent and uniformly distributed symbols. In such a case, the size of the generated observation is $k$ information bits. The monitored data is transmitted through a noisy channel to a central location. We use the notation $\mathbb{N}$ for non-negative integers and the notation $\mathbb{Z}^+$ for positive integers. Also, we define the notation $\left[m\right]$ as $\left[m\right]\triangleq\left\{1,2,\cdots,m\right\}$ for any positive integer $m \in \mathbb{Z}^+$.
	
	\begin{figure}[t] 
		\centering
		\includegraphics[angle=0,width=0.8\textwidth]{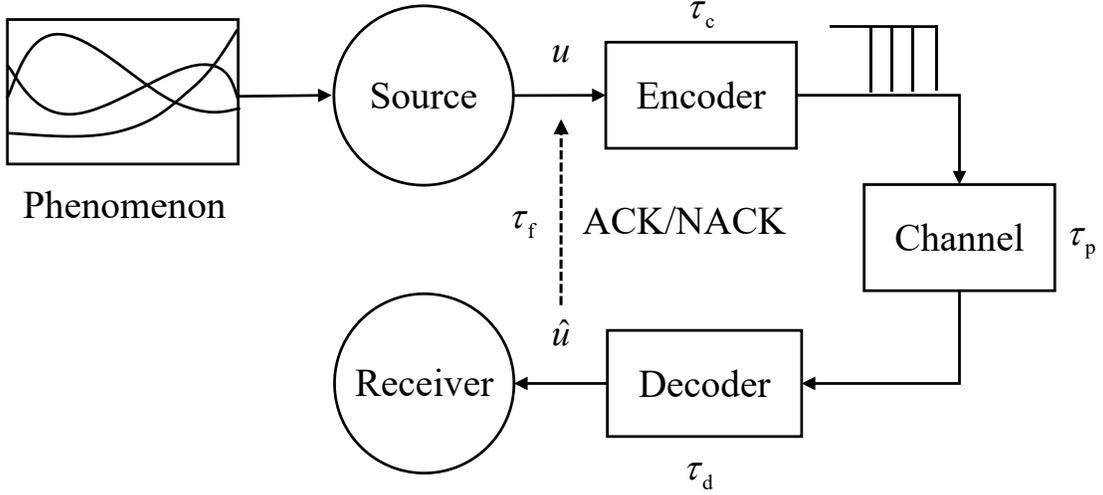}
		\caption{A general HARQ-based real-time status update system.} 
		\label{systemmodel}
	\end{figure}
	\subsection{Channel Model}
	We consider an E2E communication setup leveraging a power-limited AWGN model:
	\begin{equation}
	Y=\sqrt{P}X+Z,
	\end{equation}
	where $P$ is the average transmit power, $X$ is the unit-variance coded symbol and $Z \sim \mathcal{N}\left(0,1\right)$ is the independent and identically distributed (i.i.d) AWGN.
	\begin{figure}[t] 
		\centering  
		\subfigtopskip=2pt 
		\subfigbottomskip=2pt 
		\subfigcapskip=-5pt 
		\subfigure[Reactive HARQ]{
			\label{reactive}
			\includegraphics[angle=0,width=0.48\textwidth]{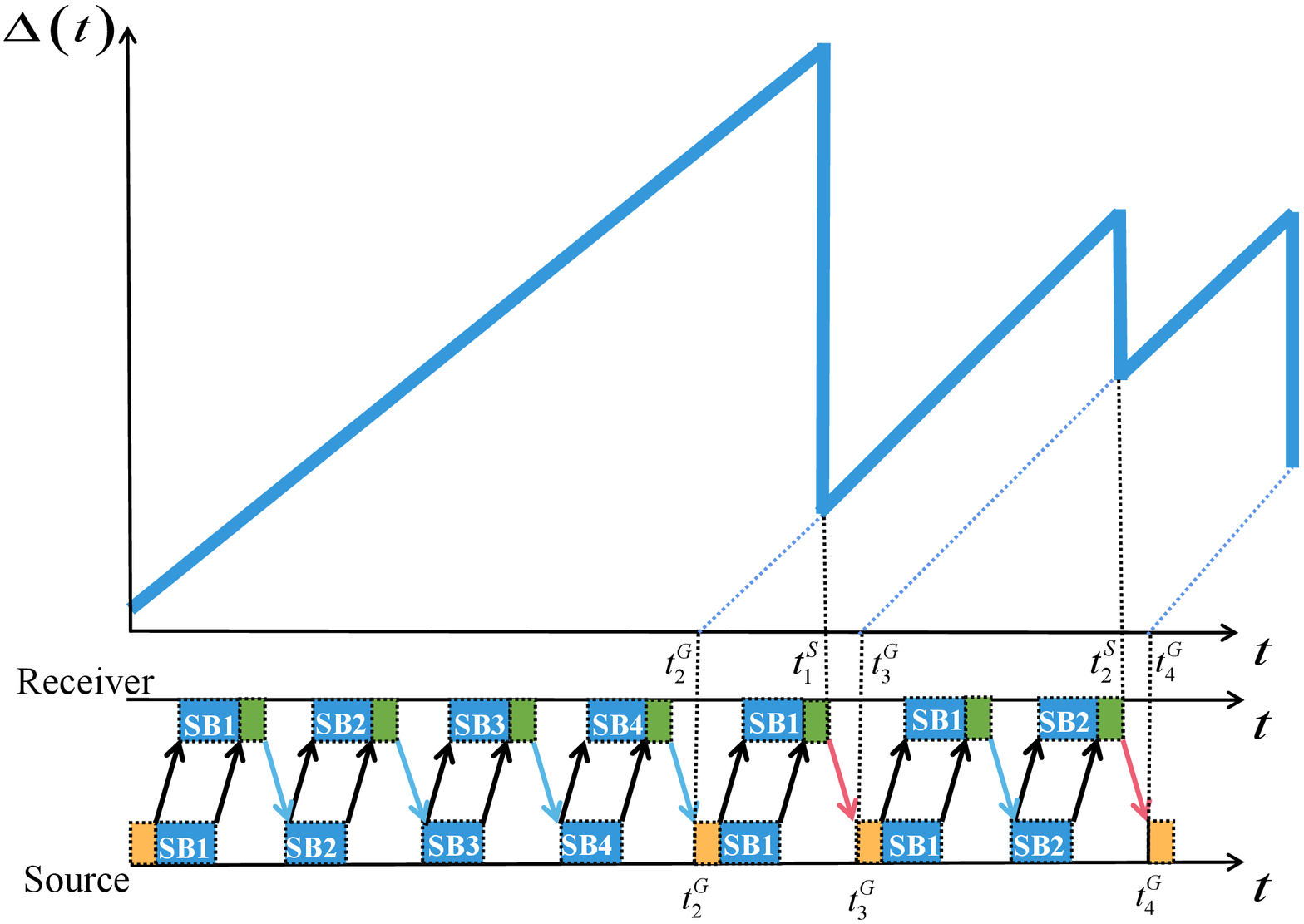}}
		
		\subfigure[Proactive HARQ]{
			\label{proactive}
			\includegraphics[angle=0,width=0.48\textwidth]{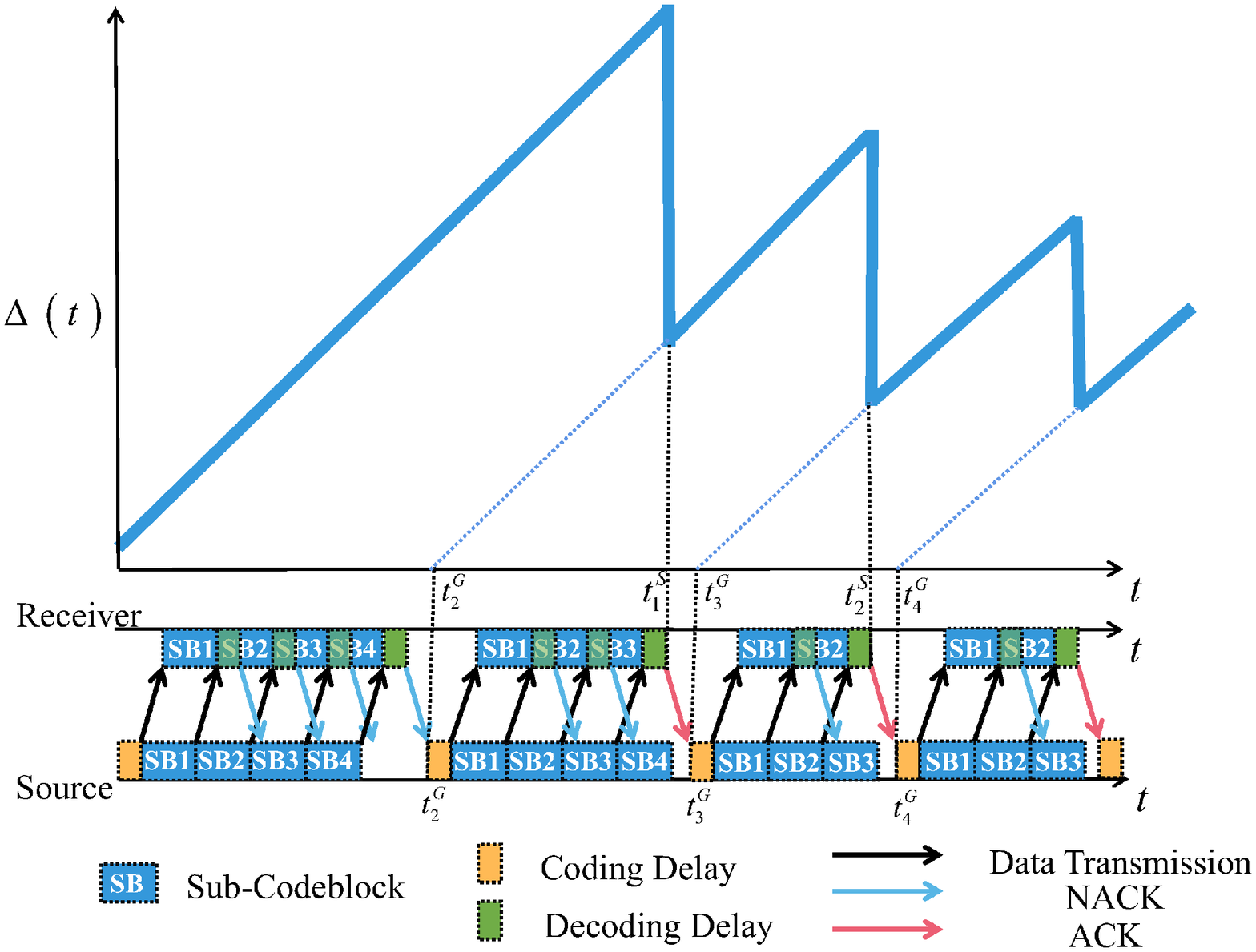}}
		
		\caption{Instantaneous age evolutions of reactive HARQ and proactive HARQ. Here the maximum retransmissions (the number of sub-codeblocks) is set as $m=4$. The length of the yellow block represents the coding delay $\tau_{\rm c}$, the length of the green block denotes the decoding delay $\tau_{\rm d}$, the length of the blue block represents the transmission delay $\ell_i$, the length of the oblique arrow projective on the timeline describes the propagation delay $\tau_{\rm p}$ (or feedback delay $\tau_{\rm f}$).}
		\label{protocols}
	\end{figure}
	\begin{remark}
		We notice that the AoI analysis over erasure channels has been extensively studied in the existing literature. However, the performance over AWGN channels is still not clear so far. As such, we consider the AWGN channel in this paper to reveal the AoI performance over the AWGN channel. 
	\end{remark}
	\subsection{Hybrid ARQ}
	
	The overall system model is shown in Fig. \ref{systemmodel}. The considered system is in close-up fashion with perfect HARQ feedback\footnote{In this article, we assume that the feedback is error-free. The research with erroneous feedback can be extended following this work.}.  At the transmitter end, the update generator (source) generates a $k$-bit short-packet update $u$ and encodes it to a \emph{parent codeword} with length $\sum\nolimits_{i = 1}^m {{\ell _i}}$ channel uses, which is then divided into $m$ sub-codeblocks with length $\ell_i, i\in\left[m\right]$ and stored in a buffer waiting to be transmitted. The coding process above will take up $\tau_{\rm c}$ channel uses, and we call it as coding delay. Next, the stored sub-codeblocks are transmitted over a noisy channel sub-codeblock by sub-codeblock, with each transmission taking a transmission delay $\ell_i, i\in\left[m\right]$ channel uses. The transmitted sub-blocks will take $\tau_{\rm p}$ channel uses to arrive at the receiver end. At the receiver end, we assume that the decoding process is conducted once receiving any complete sub-codeblock. The decoding delay in each transmission round is assumed to be the same and is denoted by $\tau_{\rm d}$. If the update is decoded correctly such that $\hat{u}=u$, an ACKnowledgment (ACK) will be fed back to the transmitter; otherwise, a NACK will be sent back. The feedback, similar to the forwarding information propagation, generally takes time and results in delay by nature, and we denote the delay as $\tau_{\rm f} $ channel uses. 
	
	
	The \emph{generate-at-will} model is adopted in the considered E2E status update system. That is, when the transmitter receives an ACK, the process of sensing and sampling will be performed, and a new update will be generated. In such a case, we mainly focus on two types of HARQ schemes: \emph{reactive} HARQ and \emph{proactive} HARQ. The detailed processes are shown in Fig. \ref{reactive} and Fig. \ref{proactive}, respectively.
	\subsubsection{Reactive HARQ}
	Reactive HARQ is also know as stop-and-wait HARQ. In Fig. \ref{reactive}, we demonstrate a detailed stop-and-wait retransmission process of \emph{reactive} HARQ, wherein the maximum number of sub-codeblocks (or maximum retransmissions) is set as $m=4$. The so-called \emph{reactive} scheme implies that the transmitter allows for retransmissions only upon the reception of a NACK. As such, the transmitter should always wait for a feedback to decide whether to generate a new update or retransmit the old update's sub-codeblocks. The waiting time, however, is referred ro as the HARQ round trip time (RTT) and will result in additional latency. Therefore, the \emph{reactive} HARQ scheme allows for only a limited number of retransmissions in the URLLC application scenarios and thereby enables great potential to be further advanced \cite{GFHARQ3}. 
	
	\subsubsection{Proactive HARQ}
	The \emph{proactive} HARQ scheme with maximum sub-codeblocks $m=4$ is shown in Fig. \ref{proactive}. As its name indicates, the retransmission process is completely spontaneous and proactive, which is interrupted only when an ACK is received. The core idea of \emph{proactive} HARQ is to eliminate the need for waiting for a feedback and implement consecutive retransmitting. By \emph{proactive} retransmitting, the latency introduced by waiting for a feedback is reduced, and thus the issue of long HARQ round trip time (RTT) is resolved.


	\subsection{Performance Metric}
	We focus on AoI analysis and optimization in this paper. Here we simply review the definition of instantaneous AoI as in Definition \ref{def1}. For more intuitive and visualized results, Fig. \ref{protocols} also gives the instantaneous age evolutions of reactive HARQ and proactive HARQ respectively. 
	\begin{definition}\label{def1}
		\textbf{(AoI)} { Denote $t_i^{\rm G}$ as the generation time instant of the $i^{\rm th}$ status update packet that can be correctly decoded, and denote $t_i^{\rm S}$ as the time instant at which this packet is correctly decoded. At a time instant $\tau$, the index of the most recently generated update can be given by $N(\tau)=\min\left\{i|t_i^{\rm S}>\tau \right\}$ and the time stamp is $U(\tau)=t_{N(\tau)}^{\rm G}$. Then, the instantaneous AoI is defined as $\Delta\left(t\right) \triangleq t-U\left(t\right)$.}
	\end{definition} 
	
	\subsubsection{Average AoI}
	Also known as time-average AoI, Average AoI is a statistical metric that measures the long-term average age of a status update system. In our considered discrete symbol-level system where the time is divided into some time slots in units of channel use, the average AoI is defined as follows.
	\begin{definition}
		\textbf{(Average AoI)} The average AoI of a real-time status update system is defined as:
		\begin{equation}
		\bar{\Delta} \triangleq \mathop {\lim }\limits_{N \to \infty } \frac{1}{N}\sum\nolimits_{t = 1}^N {\Delta \left( t \right)}.
		\end{equation}
	\end{definition}
	
	\subsubsection{Average Peak AoI }
	We also provide explicit expressions for the average Peak AoI in this paper. The peak age indicates the maximum value of age in each renewal process. In our considered system, the average Peak AoI is defined as follows.
	\begin{definition}\label{peakaoidef}
		\textbf{(Average Peak AoI)} The average Peak AoI of a real-time status update system is
		\begin{equation}
		\bar{\Delta}^{\rm P}\triangleq\lim\limits_{N \to \infty }\frac{1}{N}\sum_{j=1}^{N}\Delta\left(t^S_j-1\right).
		\end{equation}
	\end{definition}
	
	\section{Analytical Results}\label{section3}
	In this section, we study the symbol-level AoI of reactive HARQ and proactive HARQ. We first give the closed-form expressions for the AoI in Proposition \ref{reactive AoI} and Proposition \ref{proactive AoI}, and then conduct a theoretical AoI comparison between the two considered transmission protocols in Corollary \ref{coro1}. The AoI expressions, given in (\ref{eqreactive}) and (\ref{eqproactive}), are functions of the block assignment vector $\bf n$ and its dependent error probability vector $\bf e$, where the element $n_i$ in vector $\bf n$ denotes the number of cumulative transmitted symbols up to the $i^{\rm th}$ transmission round with $n_i= \sum\nolimits_{j = 1}^i {{\ell _j}}$, and the element $\epsilon_i$ in vector $\bf e$ denotes the probability that the $i^{\rm th}$ re-transmitted message remains incorrectly decoded. 
	
	By flexible choices of the vector $\bf n$ and the vector $\bf e$, we also demonstrate that our derived expression for reactive transmission protocol also unifies the available AoI analyses in the existing literature. Moreover, by using the result of the achievable rate of finite-length codes, we can obtain the AoI closed-form expression under the finite block-length (FBL) regime.
	\subsection{Reactive Scheme}
	\subsubsection{Average AoI}
	Denote the generation time of the $j^{\rm th}$ collected message as $t_j$, and denote the $j^{\rm th}$ collected message as $M_j\triangleq\mathcal{F}\left(t_j\right)$. The $j^{\rm th}$ collected message $M_j$ is encoded to a \emph{parent code} $\mathcal{C}\left(M_j\right)$ with size $n_m=\sum\nolimits_{a\in[m]}\ell_a$ being stored in a transmission buffer. Then, the transmission is evoked round by round until all the symbols stored in the buffer is transmitted or an ACK is received. In the $i^{\rm th}$ round of transmission, the transmitter will transmit $\ell_i$ symbols, and the decoder will leverage the cumulatively received $n_i={\sum\nolimits_{a\in[i]}}\ell_a$ symbols to decode the message. 
	
	In such a case, we introduce $\varsigma _{j,i} \in \left\{0,1\right\}$ to denote the feedback signal in the $i^{\rm th}$ transmission round when transmitting message $M_j$: If $\varsigma_{j,i}=1$ or $i=m$, the transmitter will no longer transmit message $M_j$, instead, it will collect new update $M_{j+1}$ for transmission; If $\varsigma_{j,i}=0$ and $i<m$, the source will transmit additional $\ell_{i+1}$ encoded symbols from $\mathcal{C}\left(M_j\right)$. As such, the probability that the $i^{\rm th}$ transmission is not correctly decoded is given as $\epsilon_i=1-\mathbb{E}\varsigma_{j,i}$.
	
	Let $Q_j\triangleq\inf\left\{i>Q_{j-1}:\varsigma_{i,m}=1\right\}, Q_0=0$ be the cumulative number of generated packets until the $j^{\rm th}$ decoding success. Let $R_j\triangleq Q_j-Q_{j-1}-1$ represent the decoding failures between two successful decoding and $V_j\triangleq\inf\left\{i\in [m]:\varsigma_{Q_j,i}=1 \right\}$ denote the round in which the $Q_j^{\rm th}$ packet gets decoded, we have that
	\begin{lemma}\label{lemma1}
		The random sequences $R_j$ and $V_j$ are independent, and they are i.i.d with distributions
		\begin{equation}\label{distribution}
		\begin{aligned}
		&\mathbb{P}\left(R_j=a\right)=\left(1-\epsilon_m\right)\epsilon^a_m, \quad a\in\mathbb{N},\\
		&\mathbb{P}\left(V_j=i\right)=\frac{\epsilon_{i-1}-\epsilon_i}{1-\epsilon_m}, \quad i\in[m].
		\end{aligned}
		\end{equation}
	\end{lemma}
	\begin{proof}
		Please refer to Appendix \ref{appendixa}.
	\end{proof}
	\begin{lemma}\label{lemma2}
		The first and second moments for $R_j$ are
		\begin{equation}\label{momentsRj}
		\begin{aligned}
		\mathbb{E}R_j&=\frac{\epsilon_m}{1-\epsilon_m},\\ \mathbb{E}R_j^2&=\frac{\epsilon_m^2+\epsilon_m}{\left(1-\epsilon_m\right)^2}.
		\end{aligned}
		\end{equation}
	\end{lemma}
	\begin{proof}
		By utilizing the distributions in (\ref{distribution}) and the definitions that $\mathbb{E}R_j=\sum_{a\in \mathbb{N}}a\cdot\mathbb{P}\left(R_j=a\right)$ and $\mathbb{E}R_j^2=\sum_{a\in \mathbb{N}}a^2\cdot\mathbb{P}\left(R_j=a\right)$, we obtain the results in (\ref{momentsRj}). 
	\end{proof}
	We then let $\tau^{\rm Reac}_i$ represent the elapsed time form generating an update to receiving its $i^{\rm th}$ feedback signal for reactive HARQ. For reactive HARQ, the transmitter should wait for an integral RTT to receive a feedback. As such, we have $\tau^{\rm Reac}_{i}=n_i+\tau_{\rm c}+i\left(\tau_{\rm d}+\tau_{\rm f}+\tau_{\rm p}\right)$. Denote $\mathcal{T}=\tau_{\rm d}+\tau_{\rm f}+\tau_{\rm p}$, the first and second moments for $\tau^{\rm Reac}_{V_j}$ is given as follows.
	\begin{lemma}\label{lemma3}
		The first and second moments for $\tau^{\rm Reac}_{V_j}$ are
		\begin{equation}\label{momentstsouR}
		\begin{aligned}
		\mathbb{E}\tau^{\rm Reac}_{V_j}&=n_1+\tau_{\rm c}+\mathcal{T}+\sum_{i\in [m-1]}\frac{\left(n_{i+1}-n_i+\mathcal{T}\right)\epsilon_i}{1-\epsilon_m},\\ 
		\mathbb{E}\left(\tau^{\rm Reac}_{V_j}\right)^2&=\left(n_1+\tau_{\rm c}+\mathcal{T}\right)^2+\sum_{i\in [m-1]}\frac{\left(n_{i+1}-n_i+\mathcal{T}\right)\left(n_{i+1}+n_i+2\tau_{\rm c}+\left(2i+1\right)\mathcal{T}\right)\epsilon_i}{1-\epsilon_m}.
		\end{aligned}
		\end{equation}
	\end{lemma}
	\begin{proof}
		By utilizing the distributions in (\ref{distribution}) and the definitions that $\mathbb{E}\tau^{\rm Reac}_{V_j}=\sum_{i\in [m]}\tau^{\rm Reac}_i\cdot\mathbb{P}\left({V_j}=i\right)$ and $\mathbb{E}\left(\tau^{\rm Reac}_{V_j}\right)^2=\sum_{i\in [m]}\left(\tau^{\rm Reac}_i\right)^2\cdot\mathbb{P}\left({V_j}=i\right)$, we obtain the results in (\ref{momentstsouR}). 
	\end{proof}
	With these notations, we can recursively write the time-instant of the $j^{\rm th}$ successful decoding $t_j^S$ as follows.
	\begin{equation}
	t_j^S=t_{j-1}^S+\tau^{\rm Reac}_{m}R_j+\tau^{\rm Reac}_{V_j}+\tau_{\rm f}.
	\end{equation}
	Therefore, the interval between the ${\left(j-1\right)}^{\rm th}$ and the $j^{\rm th}$ successful decoding for reactive HARQ is given as follows.
	\begin{equation}\label{ReacTj}
	\begin{aligned}
	T^{\rm Reac}_j&=t_j^S-t_{j-1}^S=\tau^{\rm Reac}_{m}R_j+\tau^{\rm Reac}_{V_j}+\tau_{\rm f}.
	\end{aligned}
	\end{equation}
	Since both the random sequences $R_j$ and $V_j$ are i.i.d. and have finite first and second moments, it turns out that the sequence $T^{\rm Reac}_j$ is also i.i.d. with finite
	finite first and second moments. We give the first and second moments of $T^{\rm Reac}_j$ as follows.
	\begin{lemma}
		The first and second moments for $T^{\rm Reac}_j$ are
		\begin{equation}\label{momentsTj}
		\begin{aligned}
		&\mathbb{E}T^{\rm Reac}_j=\frac{n_1+\tau_{\rm c}+\mathcal{T}+\sum_{i\in [m-1]}{\left(n_{i+1}-n_i+\mathcal{T}\right)\epsilon_i}}{1-\epsilon_m},\\
		&\mathbb{E}\left(T^{\rm Reac}_j\right)^2=\frac{\left(n_m+\tau_{\rm c}+m\mathcal{T}\right)^2\left(1+\epsilon_m\right)}{\left(1-\epsilon_m\right)^2}\\
		&-\sum_{i\in [m-1]}\frac{1-\epsilon_i}{1-\epsilon_m}\left(n_{i+1}-n_i+\mathcal{T}\right)\left[n_i+n_{i+1}+2\tau_{\rm c}+\left(2i+1\right)\mathcal{T}+\frac{2\left(n_m+\tau_{\rm c}+m\mathcal{T}\right)\epsilon_m}{1-\epsilon_m}\right].
		\end{aligned}
		\end{equation}
	\end{lemma}
	\begin{proof}
		From (\ref{ReacTj}), we obtain
		\begin{equation}\label{momentsTjreac}
		\begin{aligned}
		&\mathbb{E}T^{\rm Reac}_j=\tau^{\rm Reac}_{m}\mathbb{E}R_j+\mathbb{E}\tau^{\rm Reac}_{V_j}+\tau_{\rm f},\\
		&\mathbb{E}\left(T^{\rm Reac}_j\right)^2=\mathbb{E}\left(\tau^{\rm Reac}_{m}R_j+\tau^{\rm Reac}_{V_j}+\tau_{\rm f}\right)^2.\\
		\end{aligned}
		\end{equation}
		Note that the random variables $R_j$ and ${V_j}$ are independent with each other and the first and second moments of $R_j$ and $\tau^{\rm Reac}_{V_j}$ have been given in Lemma \ref{lemma2} and Lemma \ref{lemma3}. Substitute them into (\ref{momentsTjreac}), we can obtain the results in (\ref{momentsTj}).
	\end{proof}
	As Definition \ref{def1} indicates, the generation time instant $U\left(\tau\right)$ is given by $U\left(\tau\right)=t^G_{N\left(\tau\right)}$. Because $N\left(t^S_{j-1}\right)=\min\left\{i|t_i^{\rm S}>t^S_{j-1} \right\}=j$, we obtain that
	\begin{equation}\nonumber
	U\left(t^S_{j-1}\right)=t^G_{j}=t^S_{j-1}+\tau_{\rm f}-\tau^{\rm Reac}_{V_{j-1}}.
	\end{equation}
	Thus, for any time slots $t$ in the $j^{\rm th}$ renewal interval $\mathcal{I}_j\triangleq\left\{t^S_{j-1},\cdots,t^S_{j}-1\right\}$, we have
	\begin{equation}
	U\left(t\right)=U\left(t^S_{j-1}\right)=t^S_{j-1}+\tau_{\rm f}-\tau^{\rm Reac}_{V_{j-1}},  \text{  for } t \in \mathcal{I}_j.
	\end{equation}
	As such, the instantaneous age is given as
	\begin{equation}\label{deltreact}
	\begin{aligned}
	\Delta\left(t\right)&=t-U\left(t\right)\\
	&=t-t^S_{j-1}-\tau_{\rm f}+\tau^{\rm Reac}_{V_{j-1}}, \text{  for } t \in \mathcal{I}_j .
	\end{aligned}
	\end{equation}
	\begin{lemma}\label{lemma5}
		For the considered reactive HARQ model, the average AoI can be calculated by
		\begin{equation} \label{ReacAoI}
		\bar{\Delta}_{\rm Reactive}=\frac{\mathbb{E}\sum_{t\in \mathcal{I}_j}\Delta{\left(t\right)}}{\mathbb{E}T^{\rm Reac}_j}=\frac{\mathbb{E}T^{\rm Proac}_j}{2\mathbb{E}\left(T^{\rm Reac}_j\right)^2}+\mathbb{E}\tau^{\rm Reac}_{V_j}-\tau_{\rm f}-\frac{1}{2}.
		\end{equation}
	\end{lemma}
	\begin{proof}
		Please refer to Appendix \ref{appendixb}.
	\end{proof}
	Then, by adopting the available first and second moments of $T^{\rm Reac}_j$ and $\tau^{\rm Reac}_{V_j}$ given in (\ref{momentstsouR}) and (\ref{momentsTjreac}), we have the average AoI for reactive HARQ as follows.
	\begin{proposition} \label{reactive AoI}
		{ \textbf{(The Generalized Closed-form Average AoI Expression for Reactive Scheme)} For reactive HARQ with maximum retransmissions $m$, block assignment vector ${\bf n}=\left(n_1,n_2,\cdots,n_m\right)$ and error probability vector ${\bf e}=\left(\epsilon_1,\epsilon_2,\cdots,\epsilon_m\right)$, the average AoI can be calculated by}	
		\begin{equation} \label{eqreactive}
		\begin{aligned}
		{\bar{\Delta} _{\rm Reactive}} =  - \frac{1}{2} - {\tau _{\rm f}} + \frac{{{\tau _{\rm c}} + {n_1} + \mathcal{T} + \sum\limits_{i = 1}^{m - 1} {\left( {{n_{i + 1}} - {n_i} + \mathcal{T}} \right){\epsilon_i}} }}{{1 - {\epsilon_m}}} + \\
		\frac{{{{\left( {{\tau _{\rm c}} + {n_1} + \mathcal{T}} \right)}^2} + \sum\limits_{i = 1}^{m - 1} {\left( {{n_{i + 1}} - {n_i} + \mathcal{T}} \right)\left( 2{{\tau _{\rm c}} + {n_{i + 1}} + {n_i} + \left( {2i + 1} \right)\mathcal{T}} \right){\epsilon_i}} }}{{2\left( {{\tau _{\rm c}} + {n_1} + \mathcal{T} + \sum\limits_{i = 1}^{m - 1} {\left( {{n_{i + 1}} - {n_i} + \mathcal{T}} \right){\epsilon_i}} } \right)}},
		\end{aligned}
		\end{equation}	
		where $\mathcal{T}=\tau_{\rm f}+\tau_{\rm p}+\tau_{\rm d}$ with $\tau _{\rm c}$, $\tau _{\rm p}$, $\tau _{\rm d}$ and $\tau _{\rm f}$ denoting the coding delay, propagation delay, decoding delay and feedback delay, respectively.
	\end{proposition}

	\subsubsection{Case Study: A Unified Result}
	With Proposition \ref{reactive AoI} in hand, we can conduct some case studies by flexibly considering the choices of the block assignment vector $\bf n$ and the error probability vector $\bf e$. By this means, we theoretically show that the closed-form AoI expressions given in this paper is a unified result.\footnote{The average AoI expressions in Case 1, Case 2 and Case 3 are corresponding to Proposition 1, Proposition 3 and Proposition 2 of \cite{Age_LDPC}, respectively. The average AoI expression in Case 4 is a variant of the result in \cite{Blockoptim}.}. Though the given examples are not exhaustive in this paper, we can observe from these case studies that the unified expression given in (\ref{eqreactive}) enables potential for exploring the intrinsic relationship and comparative insights among different types of transmission protocols.
	
	\begin{case}
		{\textbf{(Average AoI for Fixed-rate Codes without ARQ)} We show that the available average AoI expression for fixed-rate codes in \cite{Age_LDPC} is a specific case of our unified result in (\ref{eqreactive}). For fixed-rate codes without ARQ, the maximum retransmissions turns to $m=1$. Substitute $m=1$ into (\ref{eqreactive}) and remove the effect of delay elements such that $\tau_{\rm c}=\tau_{\rm p}=\tau_{\rm d}=\tau_{\rm f}=0$, we can obtain the average AoI as the Proposition 1 in \cite{Age_LDPC}: }  	
		\begin{equation} 
		\nonumber
		{\bar{\Delta} _{\rm Non-ARQ}} =  - \frac{1}{2} + \frac{{{n_1}}}{{1 - {\epsilon_1}}} + \frac{{{n_1}}}{2},
		\end{equation}
		where $n_1$ is the code length and $\epsilon_1$ is the error probability of the fixed-rate codes. 
	\end{case}

	\begin{case}
		{ \textbf{(Average AoI for Truncated ARQ (TARQ))} We demonstrate that the average AoI expression for TARQ is also a specific case of our unified result in (\ref{eqreactive}). For truncated ARQ, the transmitter retransmits the same packet till the allowable maximum retransmissions $m$ is reached or this packet is successfully received. Since the retransmitted packet is the same as the first packet, the cumulative transmitted message length is $n_i=in_1$ and the corresponding error probability is $\epsilon_i={\epsilon_1}^i$. Then, by substituting them back into (\ref{eqreactive}) and similarly remove the effect of delay elements such that $\tau_{\rm c}=\tau_{\rm p}=\tau_{\rm d}=\tau_{\rm f}=0$, we can obtain the average AoI as the Proposition 3 in \cite{Age_LDPC}: }	
		\begin{equation} 
		\nonumber
		{\bar{\Delta}_{\rm TARQ}} =  - \frac{1}{2} + {n_1}\left( {\frac{2}{{1 - {\epsilon_1}}} - \frac{1}{2} - \frac{{m\epsilon_1^m}}{{1 - \epsilon_1^m}}} \right).
		\end{equation}
	\end{case}
	
	\begin{case}
		{ \textbf{(Average AoI for Classical ARQ)} We also find that the average AoI expression for TARQ is a specific case of our unified result in (\ref{eqreactive}). For classical ARQ, the transmitter re-transmits the same packet till the packet is successfully received, while the maximum retransmissions is not limited. The classical ARQ is a special case of TARQ where $m \to \infty$. Then, by calculating the limit $\mathop {\lim }\limits_{m \to \infty } {\Delta _{\rm TARQ}}$, we can obtain the average AoI as the Proposition 2 in \cite{Age_LDPC}: }	
		\begin{equation} \nonumber
		{\bar{\Delta} _{\rm Classical - ARQ}} =  - \frac{1}{2} + {n_1}\left( {\frac{2}{{1 - {\epsilon_1}}} - \frac{1}{2}} \right).
		\end{equation}
	\end{case}
	
	\begin{case}
		{ \textbf{(Average AoI for HARQ-IR)} By removing the effect of delay elements such that $\tau_{\rm c}=\tau_{\rm p}=\tau_{\rm d}=\tau_{\rm f}=0$, we find that the result in (\ref{eqreactive}}) is transformed to a variant of that in \cite{Blockoptim}, given as:
		\begin{equation} \nonumber
		\bar{\Delta} _{\rm HARQ-IR} =  - \frac{1}{2} + \frac{{{n_1} + \sum\limits_{i = 1}^{m - 1} {\left( {{n_{i + 1}} - {n_i}} \right){\epsilon _i}} }}{{1 - {\epsilon _m}}} + \frac{{n_1^2 + \sum\limits_{i = 1}^{m - 1} {\left( {n_{i + 1}^2 - n_i^2} \right){\epsilon _i}} }}{{2\left( {{n_1} + \sum\limits_{i = 1}^{m - 1} {\left( {{n_{i + 1}} - {n_i}} \right){\epsilon _i}} } \right)}}.
		\end{equation}
	\end{case}

	\subsubsection{Average Peak AoI}
	Definition \ref{peakaoidef} has indicated that $\bar{\Delta}^{\rm P}\triangleq\lim\limits_{N \to \infty }\frac{1}{N}\sum_{j=1}^{N}\Delta\left(t^S_j-1\right)$. Note that for reactive scheme, the terms in the summation can be obtained from (\ref{deltreact}), given as $\Delta\left(t^S_j-1\right)=T^{\rm Reac}_j-1-\tau_{\rm f}+\tau^{\rm Reac}_{V_{j-1}}$. Then, from the law of large numbers, we can obtain the following almost sure equality 
	\begin{equation}\nonumber
	\bar{\Delta}^{\rm P}_{\rm Reactive}=-1-\tau_{\rm f}+\mathbb{E}T^{\rm Reac}_j+\mathbb{E}\tau^{\rm Reac}_{V_{j-1}}.
	\end{equation}
	Then, by applying the available first moments of $T^{\rm Reac}_j$ and $\tau^{\rm Reac}_{V_j}$ given in (\ref{momentstsouR}) and (\ref{momentsTjreac}), we obtain the average Peak AoI for reactive HARQ in the following proposition.
	\begin{proposition} \label{theorem6}
		{ \textbf{(The Generalized Closed-form Average Peak AoI Expression for Reactive HARQ)} For reactive HARQ with maximum retransmissions $m$, block assignment vector ${\bf n}=\left(n_1,n_2,\cdots,n_m\right)$ and error probability vector ${\bf e}=\left(\epsilon_1,\epsilon_2,\cdots,\epsilon_m\right)$, the average Peak AoI can be calculated by}	
		\begin{equation} 
		\nonumber
		{\bar{\Delta}^{\rm P} _{\rm Reactive}} = -1 - {\tau _{\rm f}} - \frac{{\left( {{\tau _{\rm c}} + {n_m} + m\mathcal{T}} \right){\epsilon_m}}}{{1 - {\epsilon_m}}} + 
		\frac{2}{{1 - {\epsilon_m}}}\left( {{\tau _{\rm c}} + {n_1} + \mathcal{T} + \sum\limits_{i = 1}^{m - 1} {\left( {{n_{i + 1}} - {n_i} + \mathcal{T}} \right){\epsilon_i}} } \right).
		\end{equation}	
	\end{proposition}
	%

	
	\subsection{Proactive Scheme}
	\subsubsection{Average AoI}
	Let $\tau^{\rm Proac}_i$ represent the elapsed time form generating an update to receiving its $i^{\rm th}$ feedback signal for proactive HARQ, we can observe from Fig. \ref{proactive} that $\tau^{\rm Proac}_{i}=n_i+\tau_{\rm c}+\mathcal{T}$. As such, the first and second moments for $\tau^{\rm Proac}_{V_j}$ is correspondingly given as follows.
	
	\begin{lemma}\label{proaclemma1}
		The first and second moments for $\tau^{\rm Proac}_{V_j}$ are
		\begin{equation}\label{momentstouPr}
		\begin{aligned}
		\mathbb{E}\tau^{\rm Proac}_{V_j}&=n_1+\tau_{\rm c}+\mathcal{T}+\sum_{i\in [m-1]}\frac{\left(n_{i+1}-n_i\right)\epsilon_i}{1-\epsilon_m},\\ 
		\mathbb{E}\left(\tau^{\rm Proac}_{V_j}\right)^2&=\left(n_1+\tau_{\rm c}+\mathcal{T}\right)^2+\sum_{i\in [m-1]}\frac{\left(n_{i+1}-n_i\right)\left(n_{i+1}+n_i+2\tau_{\rm c}+2\mathcal{T}\right)\epsilon_i}{1-\epsilon_m}.
		\end{aligned}
		\end{equation}
	\end{lemma}
	\begin{proof}
		By utilizing the distributions in (\ref{distribution}) and the definitions that $\mathbb{E}\tau^{\rm Proac}_{V_j}=\sum_{i\in [m]}\tau^{\rm Proac}_i\cdot\mathbb{P}\left({V_j}=i\right)$ and $\mathbb{E}\left(\tau^{\rm Proac}_{V_j}\right)^2=\sum_{i\in [m]}\left(\tau^{\rm Proac}_i\right)^2\cdot\mathbb{P}\left({V_j}=i\right)$, we obtain the results in (\ref{momentstouPr}). 
	\end{proof}
	Denote the interval between the ${\left(j-1\right)}^{\rm th}$ and the $j^{\rm th}$ successful decoding for proactive HARQ as $T^{\rm Proac}_j$, we can similarly derive that
	\begin{equation}\label{Tjproac}
	T^{\rm Proac}_j=\tau^{\rm Proac}_{m}R_j+\tau^{\rm Proac}_{V_j}+\tau_{\rm f}.
	\end{equation}
	Then we have the first and second moments of $T^{\rm Proac}_j$ as follows.
	\begin{lemma}\label{proaclemma2}
		The first and second moments for $T^{\rm Proac}_j$ are
		\begin{equation}\label{momentsTjProac}
		\begin{aligned}
		&\mathbb{E}T^{\rm Proac}_j=\frac{n_1+\tau_{\rm c}+\mathcal{T}+\sum_{i\in [m-1]}{\left(n_{i+1}-n_i\right)\epsilon_i}}{1-\epsilon_m},\\
		&\mathbb{E}\left(T^{\rm Proac}_j\right)^2=\frac{\left(n_m+\tau_{\rm c}+\mathcal{T}\right)^2\left(1+\epsilon_m\right)}{\left(1-\epsilon_m\right)^2}\\
		&-\sum_{i\in [m-1]}\frac{1-\epsilon_i}{1-\epsilon_m}\left(n_{i+1}-n_i\right)\left[n_i+n_{i+1}+2\tau_{\rm c}+2\mathcal{T}+\frac{2\left(n_m+\tau_{\rm c}+\mathcal{T}\right)\epsilon_m}{1-\epsilon_m}\right].
		\end{aligned}
		\end{equation}
	\end{lemma}
	\begin{proof}
		From (\ref{Tjproac}), we obtain that
		\begin{equation}\label{momentsTjproac}
		\begin{aligned}
		&\mathbb{E}T^{\rm Proac}_j=\tau^{\rm Proac}_{m}\mathbb{E}R_j+\mathbb{E}\tau^{\rm Proac}_{V_j}+\tau_{\rm f},\\
		&\mathbb{E}\left(T^{\rm Proac}_j\right)^2=\mathbb{E}\left(\tau^{\rm Proac}_{m}R_j+\tau^{\rm Proac}_{V_j}+\tau_{\rm f}\right)^2.\\
		\end{aligned}
		\end{equation}
		Substitute (\ref{distribution}) and (\ref{momentsRj}) into (\ref{momentsTjproac}), we obtain the results in (\ref{momentsTjProac}).
	\end{proof}
	For the considered proactive HARQ model, the average AoI can be similarly expressed as
	\begin{equation} \label{ProacAoI}
	\bar{\Delta}_{\rm Proactive}=\frac{\mathbb{E}\sum_{t\in \mathcal{I}_j}\Delta{\left(t\right)}}{\mathbb{E}T^{\rm Proac}_j}=\frac{\mathbb{E}T^{\rm Proac}_j}{2\mathbb{E}\left(T^{\rm Proac}_j\right)^2}+\mathbb{E}\tau^{\rm Proac}_{V_j}-\tau_{\rm f}-\frac{1}{2}.
	\end{equation}
	Hence, applying the first and second moments given in Lemma \ref{proaclemma1} and Lemma \ref{proaclemma2} leads to the explicit expression in the following proposition.
	
	\begin{proposition} \label{proactive AoI}
		{\textbf{(The Generalized Closed-form Average AoI Expression for Proactive HARQ)}  For proactive HARQ with maximum retransmissions $m$, block assignment vector ${\bf n}=\left(n_1,n_2,\cdots,n_m\right)$ and error probability vector ${\bf e}=\left(\epsilon_1,\epsilon_2,\cdots,\epsilon_m\right)$, the average AoI can be calculated by}	
		\begin{equation} \label{eqproactive}
		\begin{aligned}
		{\bar{\Delta} _{\rm Proactive}} =  - \frac{1}{2} - {\tau _{\rm f}} + \frac{{{\tau _{\rm c}} + {n_1} + \mathcal{T} + \sum\limits_{i = 1}^{m - 1} {\left( {{n_{i + 1}} - {n_i}} \right){\epsilon_i}} }}{{1 - {\epsilon_m}}} + \\
		\frac{{{{\left( {{\tau _{\rm c}} + {n_1} + \mathcal{T}} \right)}^2} + \sum\limits_{i = 1}^{m - 1} {\left( {{n_{i + 1}} - {n_i}} \right)\left( {2{\tau _{\rm c}} + 2\mathcal{T} + {n_{i + 1}} + {n_i}} \right){\epsilon_i}} }}{{2\left( {{\tau _{\rm c}} + {n_1} + \mathcal{T} + \sum\limits_{i = 1}^{m - 1} {\left( {{n_{i + 1}} - {n_i}} \right){\epsilon_i}} } \right)}}.
		\end{aligned}
		\end{equation}
	\end{proposition}
	
	
	\subsubsection{Average Peak AoI}
	As the definition indicates, we have $\bar{\Delta}^{\rm P}\triangleq\lim\limits_{N \to \infty }\frac{1}{N}\sum_{j=1}^{N}\Delta\left(t^S_j-1\right)$. For proactive HARQ, $\Delta\left(t^S_j-1\right)$ is given by $\Delta\left(t^S_j-1\right)=T^{\rm Proac}_j-1-\tau_{\rm f}+\tau^{\rm Proac}_{V_{j-1}}$. Then, from the law of large numbers, we get the following equality as
	\begin{equation}\nonumber
	\bar{\Delta}^{\rm P}_{\rm Proactive}=-1-\tau_{\rm f}+\mathbb{E}T^{\rm Proac}_j+\mathbb{E}\tau^{\rm Proac}_{V_{j-1}}.
	\end{equation}
	Finally, by applying the available first of $T^{\rm Proac}_j$ and $\tau^{\rm Proac}_{V_j}$ given in (\ref{momentstouPr}) and (\ref{momentsTjproac}), we obtain the average Peak AoI for reactive HARQ in the following proposition.
	\begin{proposition} \label{theorem7}
		{ \textbf{(The Generalized Closed-form Average Peak AoI Expression for Proactive HARQ)} For proactive HARQ with maximum retransmissions $m$, block assignment vector ${\bf n}=\left(n_1,n_2,\cdots,n_m\right)$ and error probability vector ${\bf e}=\left(\epsilon_1,\epsilon_2,\cdots,\epsilon_m\right)$, the average Peak AoI can be calculated by}	
		\begin{equation}
		\nonumber
		{\bar{\Delta}^{\rm P} _{\rm Proactive}} =  -1 - {\tau _{\rm f}} - \frac{{\left( {{\tau _{\rm c}} + {n_m} + \mathcal{T}} \right){\epsilon_m}}}{{1 - {\epsilon_m}}} + 
		\frac{2}{{1 - {\epsilon_m}}}\left( {{\tau _{\rm c}} + {n_1} + \mathcal{T} + \sum\limits_{i = 1}^{m - 1} {\left( {{n_{i + 1}} - {n_i}} \right){\epsilon_i}} } \right).
		\end{equation}	
	\end{proposition}

	\subsubsection{Case Study: Rateless Codes }
	For rateless codes, the encoder can generate as many symbols as possible to achieve error-free transmission. As such, rateless codes can be regarded as a type of proactive HARQ with infinite code-length setup. By leveraging the obtained results regarding to proactive HARQ, the average AoI and average Peak AoI of rateless code are give in the following Propositions.
	\begin{proposition} \label{theorem5}
		{ \textbf{(The Generalized Closed-form Average AoI Expression for Rateless Codes)}  For rateless codes transmitted over a noisy channel with non-trivial coding delay $\tau_{\rm c}$, propagation delay $\tau_{\rm p}$, decoding delay $\tau_{\rm d}$ and feedback delay $\tau_{\rm f}$, the average AoI can be calculated by}	
		\begin{equation} \label{rateless}
		\begin{aligned}
		{\bar{\Delta} _{\rm Rateless}} =  - \frac{1}{2} + \mathcal{T} +{n_1} +  + \sum\limits_{i = 1}^{\infty} {\left( {{n_{i + 1}} - {n_i}} \right){\epsilon_i}}   \\
		+\frac{{{{\left( {{\tau _{\rm c}} + {n_1} + \mathcal{T}} \right)}^2} + \sum\limits_{i = 1}^{\infty} {\left( {{n_{i + 1}} - {n_i}} \right)\left( {2{\tau _{\rm c}} + 2\mathcal{T} + {n_{i + 1}} + {n_i}} \right){\epsilon_i}} }}{{2\left( {{\tau _{\rm c}} + {n_1} + \mathcal{T} + \sum\limits_{i = 1}^{\infty} {\left( {{n_{i + 1}} - {n_i}} \right){\epsilon_i}} } \right)}}.
		\end{aligned}
		\end{equation}
	\end{proposition}
	\begin{proof}
		Rateless codes is a special type of proactive HARQ where $m\to \infty$. Since $\lim\limits_{m \to \infty}\epsilon_m=0$, we can obtain the average AoI of rateless codes as in (\ref{rateless}) by calculating the limit $\mathop {\lim }\limits_{m \to \infty } {\bar{\Delta} _{\rm Proactive}}$.
	\end{proof}
	
	\begin{proposition} \label{retelesspeakaoi}
		{ \textbf{(The Generalized Closed-form Average Peak AoI Expression for Rateless Codes)}  For rateless codes transmitted over the channel with non-trivial coding delay $\tau_{\rm c}$, propagation delay $\tau_{\rm p}$, decoding delay $\tau_{\rm d}$ and feedback delay $\tau_{\rm f}$, the average AoI can be calculated by}	
		\begin{equation} \label{ratelesspeak}
		{\bar{\Delta}^{\rm P} _{\rm Rateless}} =  -1- {\tau _{\rm f}} + 
		{2}\left( {{\tau _{\rm c}} + {n_1} + \mathcal{T} + \sum\limits_{i = 1}^{\infty} {\left( {{n_{i + 1}} - {n_i}} \right){\epsilon_i}} } \right).
		\end{equation}
	\end{proposition}
	\begin{proof}
		Rateless codes is a special case of proactive HARQ where $m\to \infty$. By calculating the limit $\mathop {\lim }\limits_{m \to \infty } {\bar{\Delta} _{\rm Proactive}}$, we can obtain the average AoI of rateless codes as in (\ref{ratelesspeak}).
	\end{proof}
	\begin{remark}
		Note that there are infinite series in (\ref{rateless}) and (\ref{ratelesspeak}), which are
		\begin{equation}\label{inseries}
		\begin{aligned}
		&\sum\limits_{i = 1}^{\infty} {\left( {{n_{i + 1}} - {n_i}} \right){\epsilon_i}},\\
		&\sum\limits_{i = 1}^{\infty} {\left( {{n_{i + 1}} - {n_i}} \right)\left( {2{\tau _{\rm c}} + 2\mathcal{T} + {n_{i + 1}} + {n_i}} \right){\epsilon_i}}.
		\end{aligned}
		\end{equation}
		In this regard, the sufficient condition that Proposition \ref{theorem5} and Proposition \ref{retelesspeakaoi} exist is that the infinite series in (\ref{inseries}) converge to some finite values. In the following Lemma \ref{lemma8}, we would like to discuss this issue.
	\end{remark}
	\begin{lemma}\label{lemma8}
		The infinite series in (\ref{inseries}) are always bounded (less than or equal to some finite number).
	\end{lemma}
	\begin{proof}
		Please refer to Appendix \ref{appendixc}.
	\end{proof}
	\subsection{Reactive HARQ vs. Proactive HARQ}
	\begin{corollary} \label{coro1}
		{ \textbf{(Reactive HARQ vs. Proactive HARQ)} The average age performance of reactive HARQ would not exceed that of proactive HARQ under the same block assignment vector $\bf n$ and the same error probability vector $\bf e$. The necessary and sufficient condition for their equivalence is $m=1$ or $\mathcal{T}=0$. }		
	\end{corollary}
	\begin{proof}
		Please refer to Appendix \ref{appendixd}.
	\end{proof}
	Corollary \ref{coro1} demonstrates that $\bar{\Delta}_{\rm Reactive}\ge\bar{\Delta}_{\rm Proactive}$, where the equivalence happens only if $i$) $m=1$, in such a case, both the reactive HARQ and the proactive HARQ turns to a open-loop fashion non-ARQ system, and the system does not send any incremental redundancy; $ii$) $\mathcal{T}=0$, this condition infers to an ideal assumption where the propagation delay, decoding delay and feedback delay are negligible. In this regard, the RTT issue of reactive HARQ does not exist any more, and thus the considered reactive scheme is the same as the proactive one.
	\subsection{Average Age in the FBL Regime}
	With the above closed-form results, we observe that the average AoI of a HARQ-based system can be directly evaluated by determining the error probability vector $\textbf{e}$ and the block assignment vector $\textbf{n}$. The error probability $\epsilon_i$ is affected by three factors, which are, $i$) the channel condition; $ii$) the coding and decoding technique; $iii$) the message length $k$ and the code length $n_i$. As such, the framework given in this paper is general-purpose, enabling potential AoI research under different coding schemes and channel conditions.
	
	For instance, the given generalized expressions allow us to adopt the FBL results in \cite{FBL} to evaluate the AoI of the considered HARQ protocols. Over the power-limited AWGN channel with SNR $\gamma$, the error probability $\epsilon_i$ can be approximated by the Theorem 54 of \cite{FBL} as:\footnote{Here we focus on FBL analysis as a case study. Note that $\epsilon_i$ can also be characterized by other specific error-correcting techniques.  }
	\begin{equation} \label{Rateerr}
	{\epsilon_i} \approx Q\left( {\frac{C\left( { \gamma } \right) - k/{n_i} -\frac{1}{2}{\log_2{n_i}}/{n_i} } {{\sqrt {V\left( \gamma  \right)/{n_i}} }}} \right).
	\end{equation}
	where $C$ is the channel capacity with $C=\log_2\left(1+\gamma\right)$, $V$ is the channel dispersion with $V=(1-\frac{1}{\left(1+\gamma\right)^2})\log_2^2e$, and $Q\left(\cdot\right)$ denotes the $Q$ function with
	\begin{equation} \nonumber
	Q\left(x\right)=\int_{x}^{\infty}\frac{1}{\sqrt{2\pi}}\exp\left(-\frac{t^2}{2}\right)dt.
	\end{equation}
	
	Substitute (\ref{Rateerr}) into (\ref{eqreactive}) and (\ref{eqproactive}), we obtain the average AoI closed-form expressions in the FBL regime.

	\section{Age-optimal Block Assignment}
	In addition to the error probability $\bf e$, which has been determined by the available finite-length results in (\ref{Rateerr}), the other factor that can significantly affect the age performance is the block assignment vector $\bf n$. The block assignment vector $\bf n$ is an important system parameter regarding the design of the transmission strategy. To this end, this section provides design guidelines for the system parameters selection for an status update system to improve the average information timeliness. By this means, we would like to answer how many retransmissions should be and what lengths they are in an age-optimal system.
	
	\subsection{Problem Formulation}
	\begin{algorithm}[h]
		\label{Algorithm 1}
		\caption{The algorithm for solving Problem \ref{p1}.}
		\LinesNumbered
		\KwIn{The signal-to-noise ratio $\gamma$; The message length $k$; The lower bound of the range of block length $n_{\min}$; The upper bound of the range of block length $n_{\max}$; The system delay $\tau_{\rm c}$, $\tau_{\rm p}$, $\tau_{\rm d}$ and $\tau_{\rm f}$; }
		\KwOut{The optimal block assignment vector ${\bf n}_{\rm optimal}$; The minimum average age $\bar{\Delta}_{\min}$;}
		Initialization:  $\bar{\Delta}_{\min} = \infty$; $\mathcal{S}=\left\{0,1\right\}^{n_{\max}-n_{\min}+1}$\;
		\For{${\bf p}$ in $\mathcal{S}$ }
		{
			Map vector $\bf p$ to the block assignment vector $\bf n$\;
			According to the obtained $\bf n$, calculate the average age $\bar{\Delta}$ by using  (\ref{eqreactive}) or (\ref{eqproactive}) \;
			\If{$\bar{\Delta} < \bar{\Delta}_{\min}$}
			{
				Update $\bar{\Delta}_{\min}=\bar{\Delta}$\;
				Update ${\bf n}_{\rm optimal}={\bf n}$\;
			}
		}
		\Return ${\bf n}_{\rm optimal}$ and $\bar{\Delta}_{\min}$
	\end{algorithm}
	We establish an average AoI minimization problem here to further explore the age-optimal transmission mechanism in the FBL regime with non-trivial delay:
	
	{\rm 1) Objective function: To minimize the average age $\bar{\Delta}$. }
	
	{\rm 2) Decision variable: The block assignment vector ${\bf n}=\left(n_1, n_2, \cdots, n_m\right)$.   
		
		\begin{problem} Age-optimal block assignment for reactive HARQ (or proactive HARQ)\label{p1}
			\begin{equation}\nonumber
			\begin{aligned}
			\min_{\bf n} \quad &{\bar{\Delta} _{\rm Reactive}} \quad  {\rm or} \quad  {\bar{\Delta} _{\rm Proactive}}\\
			\mbox{s.t.} \quad 
			&c_1:{{n_{\min }} \le {n_1} < {n_2} <  \ldots  < {n_m} \le {n_{\max }}},\\
			&c_2:{1 \le m \le n_{\max}-n_{\min}+1},\\
			&c_3:{\epsilon_i} = Q\left( {\frac{C\left( { \gamma } \right) - k/{n_i} -\frac{1}{2}{\log_2n_i}/{n_i} } {{\sqrt {V\left( \gamma  \right)/{n_i}} }}} \right),\\
			&c_4:{m,{n_i} \in {\mathbb{Z}^ + },i = 1, \ldots ,m}.
			\end{aligned}
			\end{equation}
		\end{problem}
		
		Note that the decision variable $\bf n$ is a variable-length vector with infinite solution space, we introduce $n_{\min}$ and $n_{\max}$ as constraints of the solution space, which denotes the lower bound and the upper bound of the range of block length, respectively.

		\subsection{Solutions and Discussions}
		Problem \ref{p1} is a nonlinear integer problem. To solve the optimal solution of Problem \ref{p1}, an auxiliary vector $ {\bf p} \in \mathcal{S} \triangleq \left\{0,1\right\}^{n_{\max}-n_{\min}+1}$ can be introduced here.
		\begin{lemma} \label{mappingnp}
			There exists an one-to-one mapping between vectors $\bf n$ and $\bf p$.
		\end{lemma} 
		\begin{proof}
			Please refer to appendix \ref{appendixe}, where we construct a specific one-to-one mapping function between.
		\end{proof}
		Lemma \ref{mappingnp} illustrates that the introduced auxiliary vector $\bf p$ can be regarded as an index of the solution space of Problem \ref{p1}, which can help us traverse the entire solution space efficiently and find the optimal solution. The detailed algorithm process is provided in Algorithm \ref{Algorithm 1}.
		
		\begin{figure}[t] 
			\centering
			\includegraphics[angle=0,width=0.9\textwidth]{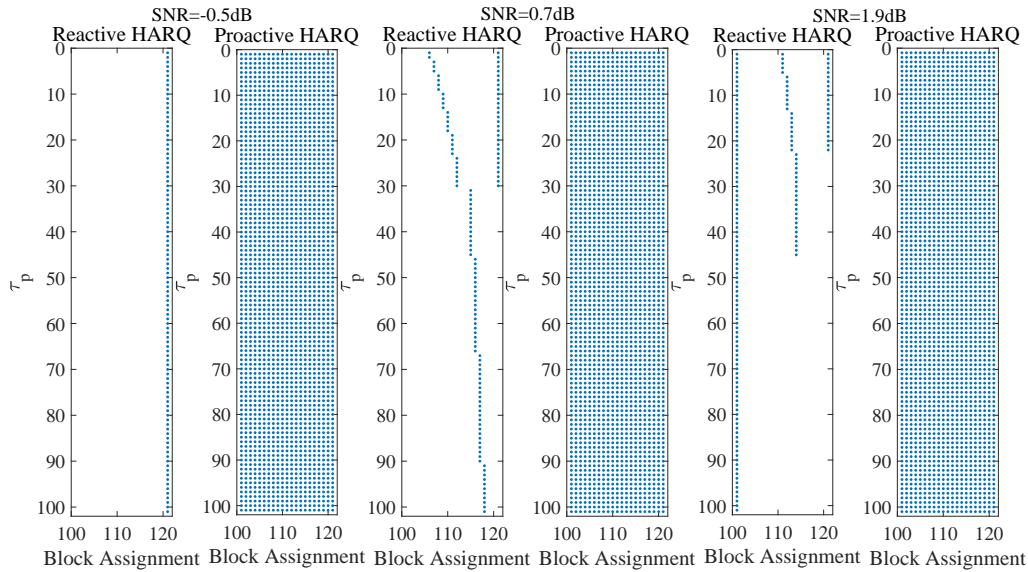}
			\caption{Age-optimal block assignment of reactive HARQ and proactive HARQ. Here the length of message is $k=100$, the minimum code length $n_{\min}=100$, the maximum code length $n_{\max}=120$, the coding delay is $\tau _{\rm c}=20$, the decoding delay is $\tau _{\rm d}=30$.} 
			\label{Age-optimal block assignment}
		\end{figure}
		
		Fig. \ref{Age-optimal block assignment} gives some detailed examples of the solved optimal block assignment vector ${\bf n}_{\rm optimal}$ under different protocols, SNRs, and propagation delays. For example, under SNR$=0.7$ dB and $\tau_{\rm p}=0$, the optimal block assignment vector for reactive HARQ is ${\bf n}_{\rm optimal}=\left(105,120\right)$; under SNR$=1.9$ dB and $\tau_{\rm p}=20$, the optimal block assignment vector for reactive HARQ is ${\bf n}_{\rm optimal}=\left(100,112,120\right)$.
		
		Fig. \ref{Age-optimal block assignment}. leads to the following conclusions:
		
		\begin{itemize}
			\item  For proactive HARQ, the finest grained symbol-by-symbol strategy always minimizes the average AoI. 
			\item  For reactive HARQ, the age-optimal block assignment varies among different SNRs and propagation delays. As the propagation delay increases, the number of retransmissions will monotonically decrease and finally converge to $m=1$. In such a case, the transmission scheme turns an open-loop fashion without any retransmission. This indicates that there exists a threshold of the propagation delay, only within which retransmission is beneficial to AoI. 
			\item From a perspective of channel coding, we can see that the trade-off between reliability and effectiveness can be well evaluated by the new metric, AoI. It is well known that a longer code length can improve the reliability while sacrificing the effectiveness; however, what is not fully explored is that an appropriate choice of code length can minimize the AoI. 
		\end{itemize}
		
		\subsection{A Heuristic Algorithm for Reactive HARQ}
		For proactive HARQ, it has been empirically shown that the finest grained strategy minimizes the average AoI; however, for reactive HARQ, the age-optimal strategies vary along with channel conditions and propagation delay. Thus, to repeatedly determine the age-optimal scheme requires amounts of calculations. It is also pertinent to note that the implementation of Algorithm \ref{Algorithm 1} aims to exhaustively search the whole solution space to find an age-optimal block assignment strategy. As such, the complexity of such an Algorithm is exponentially increasing with the range of code length $n_{\max}-n_{\min}+1$. For a broader range of code length, here we heuristically provide a sub-optimal algorithm to circumvent the high-complexity issue. Specifically, the heuristic sub-optimal algorithm is given in Algorithm \ref{Algorithm 2}.
		\begin{algorithm}
			\label{Algorithm 2}
			\caption{The sub-optimal algorithm for solving Problem \ref{p1}}
			\LinesNumbered
			\KwIn{The signal-to-noise ratio (SNR); The message length $k$; The lower bound of the range of block length $n_{\min}$; The upper bound of the range of block length $n_{\max}$; The system delay $\tau_{\rm c}$, $\tau_{\rm p}$, $\tau_{\rm d}$ and $\tau_{\rm f}$;}
			\KwOut{The optimal block assignment vector ${\bf n}_{\rm optimal}$; The minimum average age $\bar{\Delta}_{\min}$;}
			Initialization: $\bar{\Delta}_{\min}=\infty$\; 
			\For{$m \leftarrow$  $1$ to $n_{\max}-n_{\min}+1$ }
			{
				$\bar{\Delta}_{\min,m}=\infty$\;
				Construct the sub set $\mathcal{S}_m$\;
				\For{$\bf p$ in $\mathcal{S}_m$}
				{
					Map vector $\bf p$ to the block assignment vector $\bf n$\;
					According to the obtained $\bf n$, calculate the average age $\bar{\Delta}$ by using (\ref{eqreactive}) or (\ref{eqproactive}) \;
					\If{$\bar{\Delta} < \bar{\Delta}_{\min,m}$}
					{
						Update $\bar{\Delta}_{\min,m}=\bar{\Delta}$\;
						Update ${\bf n}_{{\rm optimal},m}={\bf n}$\;
					}
				}
				\If{$\bar{\Delta}_{\min,m} > \bar{\Delta}_{\min,m-1}$}
				{
					Update $\bar{\Delta}_{\min}=\bar{\Delta}_{\min,m-1}$\;
					Update ${\bf n}_{{\rm optimal}}={\bf n}_{{\rm optimal},m-1}$\;
					break\;
				}
				\Else{Update $\bar{\Delta}_{\min}=\bar{\Delta}_{\min,m}$\;
					Update ${\bf n}_{{\rm optimal}}={\bf n}_{{\rm optimal},m}$\;}
				
			}
			\Return ${\bf n}_{\rm optimal}$ and $\bar{\Delta}_{\min}$
		\end{algorithm}
		
		The design of such an algorithm is based on the empirical observation that the age-optimal $m$ is always monotonically decreasing with propagation delay (see Fig. \ref{Age-optimal block assignment}). Meanwhile, the age-optimal $m$ for reactive HARQ tends to remain small since large $m$ will lead to multiple RTT and thus result in \emph{staleness} of information. Attributed to the above factors, heuristically, we denote the minimal age under a fixed $m$ as $\bar{\Delta}_{\min,m}$ and recursively search the solution space $\mathcal{J}_m\triangleq\left\{\left(n_1, \cdots, n_i\right): i=m, n_{\min}<n_1< n_2 \cdots < n_i<n_{\max} \right\}$ with increasing value of $m$. The on-the-fly searching process will terminate only if $\bar{\Delta}_{\min,m}>\bar{\Delta}_{\min,m-1}$. By this means, the algorithm outputs $\bar{\Delta}_{\min,m-1}$ as the optimal age. In such a case, this algorithm eliminates the need for searching the sub-space $\cup_{i=m+1}^{n_{\max}-n_{\min}+1}\mathcal{S}_{i} \subseteq \mathcal{S}$, thereby bypassing the calculations required for searching for the whole solution space.
		
		Note that an auxiliary set $ \mathcal{S}_m \triangleq \left\{{\bf p} \in \left\{0,1\right\}^{n_{\max}-n_{\min}+1}:\left\|{\bf p}\right\|_1=m\right\}$ is also introduced in Algorithm \ref{Algorithm 2} to assist high-efficiency searching, where $\left\|\cdot\right\|_1$ represents the implementation of $\ell_1$-norm.
		
		\begin{lemma}\label{mappinghnp}
			There exists an one-to-one mapping between ${\bf n}\in \mathcal{J}_m$ and ${\bf p} \in \mathcal{S}_m$.
		\end{lemma}
		
		\begin{proof}
			The one-to-one mapping function between ${\bf n}\in \mathcal{J}_m$ and ${\bf p} \in \mathcal{S}_m$ is the same as that of Lemma \ref{mappingnp}, which has been discussed in Appendix \ref{appendixe} in detail. 
		\end{proof}
		
		\section{Numerical Results}
		
		\subsection{The Closed-Form Results}
		\begin{figure*}[htbp]
			\centering
			\includegraphics[angle=0,width=0.98\textwidth]{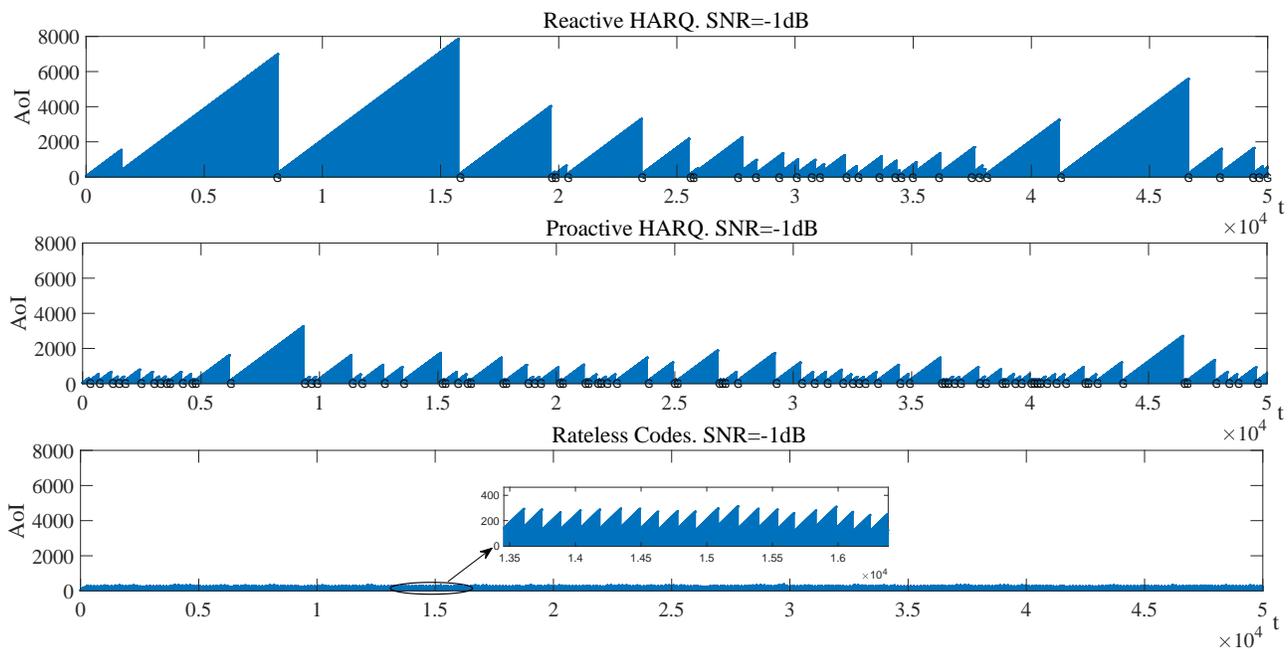}
			\caption{Instantaneous age evolution and the statistic characterizations among reactive HARQ, proactive HARQ and rateless codes. Here the encoding delay is $\tau_{\rm c}=2$, the decoding delay is $\tau_{\rm d}=3$, the propagation delay is $\tau_{\rm p}=5$ and the feedback delay is $\tau_{\rm f}=6$.   }
			\label{Simulation}
		\end{figure*}
		
		In addition to the case studies given in Section III. B, we also carry out Monte Carlo simulations to verify our closed-form expressions. For the simulation setup, we leverage an i.i.d uniformly distributed random sequence $\mathcal{X}_j\sim \mathcal{U}\left(0,1\right)$ to generate the feedback signal sequence $\varsigma_{j,i}, i\in [m]$ when transmitting message $M_j$. Specifically, the feedback signal sequence is generated by
		\begin{equation}\nonumber
		\varsigma_{j,i}=sign\left(\mathcal{X}_j-\epsilon_i\right) \text{, for } i\in[m],
		\end{equation}
		where $\epsilon_i$ is obtained by (\ref{Rateerr}) and $sign\left(\cdot\right)$ is defined as
		\begin{equation}\nonumber
		sign\left( x \right) = \left\{ 
		{\begin{array}{*{20}{c}}
			1\text{   , for } x\ge 0\\
			0\text{   , for }  x<0
			\end{array}} \right..
		\end{equation}
		\begin{figure}[t]
			\centering 
			\includegraphics[angle=0,width=0.9\textwidth]{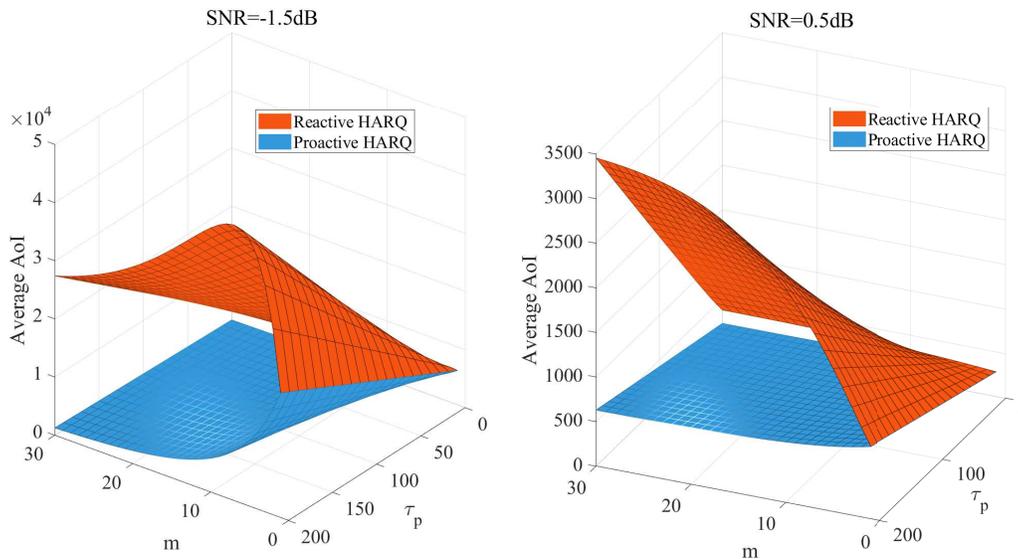}
			\caption{Reactive HARQ vs proactive HARQ. Here the message length is $k=100$, the encoding delay is $\tau_{\rm c}=20$, the decoding delay is $\tau_{\rm c}=30$, the propagation delay range is $\tau_{\rm p}=0:200$ and the feedback delay is calculated by $\tau_{\rm f}=\tau_{\rm p}+1$. For fairness, here we consider the finest grained block assignment vectors for both reactive HARQ and proactive HARQ, with $n_1=100, n_i=n_{i-1}+1$.}
			\label{revspro}
		\end{figure}
		Then, with the feedback signal sequence $\varsigma_{j,i}$ in hand, the transmission-decoding model is almost sure, and we can recursively obtain the instantaneous age evolution as shown in Fig. \ref{Simulation}. For reactive HARQ and proactive HARQ, we set $k=100$, $m=11$ and $\mathbf{n}=100:110$. For rateless codes, we find that a sufficiently large value $m$ will directly lead to an almost convergent AoI. Thus, we set $m=10000$ for the simulation setup of rateless codes.
		
		Fig. 4(a) demonstrates the instantaneous age evolution for reactive HARQ, proactive HARQ, and rateless codes, respectively. Intuitively, we can observe that the age of reactive HARQ tends to exhibit a number of large sawtooth waveforms, while that of proactive HARQ and rateless codes cut off the large sawtooth waveforms and keep at a relatively low level. 
		
		\begin{figure}[t]
			\centering
			\includegraphics[angle=0,width=0.9\textwidth]{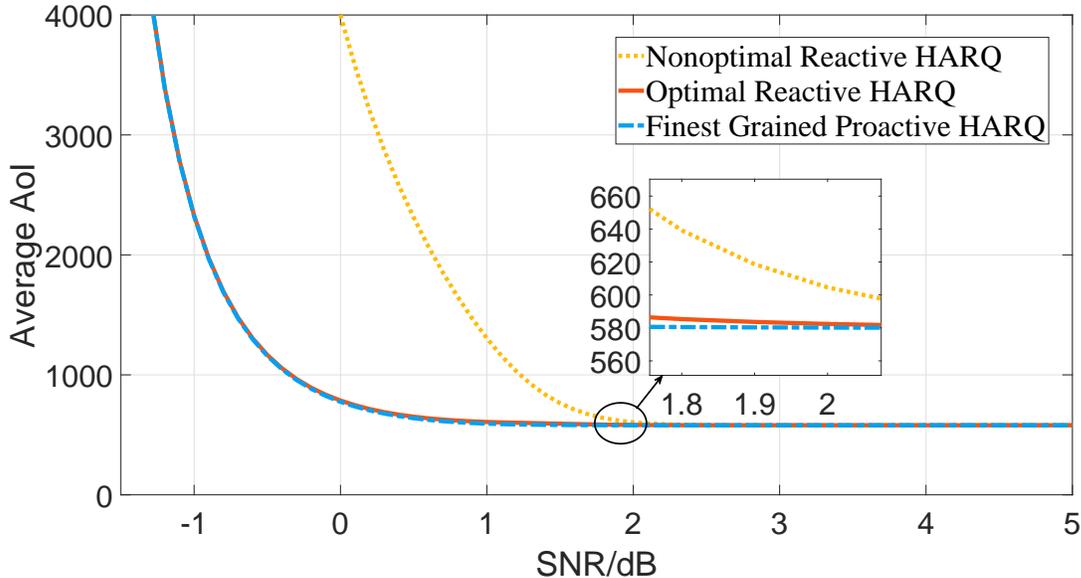}
			\caption{Average age comparisons among the finest grained reactive HARQ, the optimal reactive HARQ and the finest grained proactive HARQ. Here the message length is $k=100$, the encoding delay is $\tau_{\rm c}=20$, the decoding delay is $\tau_{\rm c}=200$, the propagation delay is $\tau_{\rm p}=50$ and the feedback delay is $\tau_{\rm f}=51$. }
			\label{blockassignment}
		\end{figure}
		
		Fig. 4(b) and Fig. 4(c) depict the average AoI and average peak AoI comparisons between the simulation results and the analytical closed-form results, wherein the discrete orange points are obtained through Monte Carlo simulations, while the blue curves are plotted by utilizing the available closed-form results given in Section III. It can be seen that the simulation results fit well with the analytical results, verifying that our provided closed-form expressions enable exact and efficient AoI evaluations.
		
		\subsection{Reactive HARQ vs. Proactive HARQ}
		
		Fig. \ref{revspro} demonstrates the average AoI comparison between reactive HARQ and proactive HARQ from a multi-dimensional perspective. The comparisons are conducted among different settings of $m$, $\tau_{\rm p}$ and SNR. It is shown that the proactive HARQ surface remains below the reactive HARQ surface. Also, they intersects with each other at $m=1$. These numerical results are consistent with Corollary \ref{coro1}. In addition, Fig. \ref{revspro} also illustrates the impact that $\tau_{\rm p}$ and $m$ exert on average AoI. On the one hand, the average AoI is monotonically increasing with respect to the propagation delay $\tau_{\rm p}$. On the other hand, the impact of $m$ on average AoI could be complex: $i$) for proactive HARQ, retransmitting redundancy remains beneficial for AoI performance metric; $ii$) for reactive HARQ, however, retransmitting redundancy naturally brings about RTT and thus results in \emph{staleness} of information when the SNR is high enough to achieve reliable communication; in contrast, if the channel condition is poor, retransmitting redundancy is essential for reliable delivery, and thus may even compensate the AoI losses due to RTT. 
		
		\subsection{The Age-Optimal Block Assignment}
		
		Fig. \ref{blockassignment} shows a comparison among the finest grained reactive HARQ, the optimal reactive HARQ, and the finest grained proactive HARQ. It is shown that the optimized reactive HARQ approaches proactive HARQ in average AoI performance. Notice that this gain lies in an adaptive block assignment strategy which requires accurate channel status information. In this regard, we find that adopting proactive HARQ for freshness-critical status update systems would be a robust and timeliness-efficient approach.

		\section{Conclusion and Future Work }
		In this paper, we have comprehensively considered different types of nontrivial system delay and derived unified closed-form average AoI and average Peak AoI expressions for both reactive HARQ and proactive HARQ. The unifying characteristic of our result has been shown by several case studies, wherein some existing PHY-layer AoI expressions in the literature are shown to be only some specific cases of our result. With these closed-form results in hand, we have theoretically proven that under the same communication conditions, the proactive scheme always outperforms the reactive scheme in terms of the average AoI. Also, a block assignment design framework at the PHY layer has been provided to further achieve timely delivery in a status update system. The simulation and analytical results demonstrate that the age-optimal block assignment strategy of reactive schemes is sensitive to both channel conditions and propagation delays, while that of proactive scheme exhibits both strategy robustness and age superiority. In this regard, we witness the potential for the proactive HARQ to be applied in the freshness-critical system. 
		
		The research in this paper also leaves some open challenges and issues for future research. First, it will be an interesting work to carry out AoI analyses and comparisons for some specific state-of-the-art channel coding techniques, such as polar codes, LDPC codes, Turbo codes, and rateless Raptor codes, etc. As such, from the AoI perspective, the trade-off among coding complexity, decoding complexity, codelength, the number of retransmitted packets, and the error probability can be explored. Second, since this work is based on an ideal assumption of perfect feedback, the analysis considering lossy feedback can be further conducted. Third, notwithstanding the AoI superiority of proactive HARQ compared to reactive HARQ, proactive HARQ may consume more energy due to the consecutive retransmissions. To this end, to further investigate the trade-off of proactive HARQ between timeliness and energy efficiency would be an interesting topic.


		\ifCLASSOPTIONcaptionsoff
		\newpage
		\fi
		\appendices
		
		\section{Proof of Lemma \ref{lemma1}} \label{appendixa}
		The event $\left\{R_j=a\right\}$ is equivalent to 
		\begin{equation}\nonumber
		\left\{R_j=a\right\}=\left\{\varsigma_{Q_{j},m}=1\right\}\bigcap_{i\in [a]}\left\{\varsigma_{i+Q_{j-1},m}=0\right\}.
		\end{equation}
		Note that the AWGN is i.i.d, we have that
		\begin{equation}\nonumber
		\begin{aligned}
		\mathbb{P}\left(R_j=a\right)&=\mathbb{P}\left(\varsigma_{Q_{j},m}=1\right)\cdot\prod_{i\in [a]}\mathbb{P}\left(\varsigma_{i+Q_{j-1},m}=0\right)\\
		&=\left(1-\epsilon_m\right)\epsilon^a_m.
		\end{aligned}
		\end{equation}
		
		Also, the probability of event $\left\{V_j=i\right\}$ can be expressed as
		\begin{equation}\label{Vj}
		\mathbb{P}\left(V_j=i\right)=\frac{\mathbb{P}\left(\left\{\varsigma_{Q_j,m}=1, \varsigma_{Q_j,i}=1 \right\}\bigcap_{r\in [i-1]}\left\{ \varsigma_{Q_j,r}=0\right\}\right)}{\mathbb{P}\left(\varsigma_{Q_j,m}=1\right)}.
		\end{equation}
		As the variable $\varsigma_{Q_j,i}$ follows the monotonic property such that 
		\begin{equation}\nonumber
		\begin{aligned}
		&\left\{\varsigma_{Q_j,i}=1\right\}\subseteq\left\{\varsigma_{Q_j,m}=1\right\}, \text{\rm for } i \le m,\\
		&\left\{\varsigma_{Q_j,i}=0\right\}\subseteq\left\{\varsigma_{Q_j,r}=0\right\}, \text{\rm for } r \le i.
		\end{aligned}
		\end{equation}
		The event in (\ref{Vj}) can be simplified as
		\begin{equation}\label{event}
		\begin{aligned}
		&\left\{\varsigma_{Q_j,m}=1, \varsigma_{Q_j,i}=1 \right\}\bigcap_{r\in [i-1]}\left\{ \varsigma_{Q_j,r}=0\right\}\\
		&=\left\{\varsigma_{Q_j,i}=1,\varsigma_{Q_j,i-1}=0\right\}\\
		&=\left\{\varsigma_{Q_j,i}=1\right\} /\left\{\varsigma_{Q_j,i-1}=1\right\}.
		\end{aligned}
		\end{equation}
		Substituting (\ref{event}) into (\ref{Vj}} results in the probability as
	\begin{equation}\nonumber
	\mathbb{P}\left(V_j=i\right)=\frac{\epsilon_{i-1}-\epsilon_i}{1-\epsilon_m}.
	\end{equation}
	
	\section{Proof of Lemma \ref{lemma5}}\label{appendixb}
	Define $G^{\rm Reac}_j=\sum_{t\in \mathcal{I}_j}\bar{\Delta}\left(t\right)$, we have that
	\begin{equation}\nonumber
	\begin{aligned}
	&G^{\rm Reac}_j=\sum_{t=t^S_{j-1}}^{t^S_{j}-1}\left(t-t^S_{j-1}-\tau_{\rm f}+\tau^{\rm Reac}_{V_{j-1}}\right)\\
	&\overset{x=t-t^S_{j-1}}{=}\sum_{x=0}^{T^{\rm Reac}_j-1}\left(x-\tau_{\rm f}+\tau^{\rm Reac}_{V_{j-1}}\right)\\
	&=\frac{T^{\rm Reac}_j\left(T^{\rm Reac}_j-1\right)}{2}-\tau_{\rm f}T^{\rm Reac}_j+\tau^{\rm Reac}_{V_{j-1}}T^{\rm Reac}_j.
	\end{aligned}
	\end{equation}
	Thus, we have the first moment of $S_j$ as
	\begin{equation}\label{Sjmoments}
	\mathbb{E}G^{\rm Reac}_j=\frac{\mathbb{E}\left(T^{\rm Reac}_j\right)^2}{2}+\mathbb{E}T^{\rm Reac}_j\left(\mathbb{E}\tau^{\rm Reac}_{V_{j-1}}-\tau_{\rm f}-\frac{1}{2}\right).
	\end{equation}
	With (\ref{Sjmoments}), we can obtain the average age as
	\begin{equation}\nonumber
	\bar{\Delta}_{\rm Reactive}=\frac{\mathbb{E}G^{\rm Reac}_j}{\mathbb{E}T^{\rm Reac}_j}=\frac{\mathbb{E}T^{\rm Proac}_j}{2\mathbb{E}\left(T^{\rm Reac}_j\right)^2}+\mathbb{E}\tau^{\rm Reac}_{V_j}-\tau_{\rm f}-\frac{1}{2}.
	\end{equation} 
	
	\section{Proof of Lemma \ref{lemma8}}\label{appendixc}
	Recall that $\epsilon_i$ a monotonically decreasing infinite sequence with $\epsilon_1>\epsilon_2>\epsilon_3>\cdots$, we can prove Lemma 8 by adopting the Dirichlet’s test. With Dirichlet’s test, this proof is equivalent to proving that the partial sums
	\begin{equation}\label{partialsum}
	\begin{aligned}
	&\sum\limits_{i = 1}^{N} {\left( {{n_{i + 1}} - {n_i}} \right)},\\
	&\sum\limits_{i = 1}^{N} {\left( {{n_{i + 1}} - {n_i}} \right)\left( {2{\tau _{\rm c}} + 2\mathcal{T} + {n_{i + 1}} + {n_i}} \right)}.
	\end{aligned}
	\end{equation}
	are bounded. Evidently, we have the solutions to (\ref{partialsum}) as
	\begin{equation}
	\begin{aligned}
	&\sum\limits_{i = 1}^{N} {\left( {{n_{i + 1}} - {n_i}} \right)}=n_{N+1}-n_1,\\
	&\sum\limits_{i = 1}^{N} {\left( {{n_{i + 1}} - {n_i}} \right)\left( {2{\tau _{\rm c}} + 2\mathcal{T} + {n_{i + 1}} + {n_i}} \right)}=\left(n_{N+1}+\tau_{\rm c}+\mathcal{T}\right)^2-\left(n_1+\tau_{\rm c}+\mathcal{T}\right)^2.
	\end{aligned}
	\end{equation} 
	Thus, we have that the infinite series in (\ref{inseries}) are bounded.
	
	\section{Proof of Corollary \ref{coro1}}\label{appendixd}
	Subtract ${\bar{\Delta} _{\rm Proactive}}$ from ${\bar{\Delta} _{\rm Reactive}}$, we have
	\begin{equation} \label{HARQvs}
	\begin{aligned}
	&{\bar{\Delta} _{\rm Reactive}} - {\bar{\Delta} _{\rm Proactive}} = \frac{{\mathcal{T}\sum\limits_{i = 1}^{m - 1} {{\epsilon_i}} }}{{1 - {\epsilon_m}}} + \\
	&\frac{{{{\left( {{\tau _{\rm c}} + {n_1} + \mathcal{T}} \right)}^2} + \sum\limits_{i = 1}^{m - 1} {\left( {{n_{i + 1}} - {n_i} + \mathcal{T}} \right)\left( {2{\tau _{\rm c}} + {n_{i + 1}} + {n_i} + \left( {2i + 1} \right)\mathcal{T}} \right){\epsilon_i}} }}{{2\left( {{\tau _{\rm c}} + {n_1} + \mathcal{T} + \sum\limits_{i = 1}^{m - 1} {\left( {{n_{i + 1}} - {n_i} + \mathcal{T}} \right){\epsilon_i}} } \right)}}\\
	&- \frac{{{{\left( {{\tau _{\rm c}} + {n_1} + \mathcal{T}} \right)}^2} + \sum\limits_{i = 1}^{m - 1} {\left( {{n_{i + 1}} - {n_i}} \right)\left( {2{\tau _{\rm c}} + 2\mathcal{T} + {n_{i + 1}} + {n_i}} \right){\epsilon_i}} }}{{2\left( {{\tau _{\rm c}} + {n_1} + \mathcal{T} + \sum\limits_{i = 1}^{m - 1} {\left( {{n_{i + 1}} - {n_i}} \right){\epsilon_i}} } \right)}}\\
	&\overset{(a)}{\ge} \frac{{\mathcal{T}\sum\limits_{i = 1}^{m - 1} {{\epsilon_i}} }}{{1 - {\epsilon_m}}} + \frac{{\mathcal{T}\sum\limits_{i = 1}^{m - 1} {\left( {{n_{i + 1}} - {n_i}} \right)\left( {2i - 1} \right){\epsilon_i}} }}{{2\left( {{\tau _{\rm c}} + {n_1} + \mathcal{T} + \sum\limits_{i = 1}^{m - 1} {\left( {{n_{i + 1}} - {n_i}} \right){\epsilon_i}} } \right)}} + \\
	&\frac{{\mathcal{T}\sum\limits_{i = 1}^{m - 1} {\left( {2{\tau _{\rm c}} + {n_{i + 1}} + {n_i} + \left( {2i + 1} \right)\mathcal{T}} \right){\epsilon_i}} }}{{2\left( {{\tau _{\rm c}} + {n_1} + \mathcal{T} + \sum\limits_{i = 1}^{m - 1} {\left( {{n_{i + 1}} - {n_i}} \right){\epsilon_i}} } \right)}}\overset{(b)}{\ge} 0.
	\end{aligned}
	\end{equation}
	
	\subsubsection{Proof of Sufficiency}
	If $\mathcal{T}=0$, the equal signs at both $(a)$ and $(b)$ in (\ref{HARQvs}) are established ; if $m=1$, all those sums $\sum\nolimits_{i=1}^{m-1} {x_i}\equiv 0 $, and thus we have ${\bar{\Delta} _{\rm Reactive}} - {\bar{\Delta} _{\rm Proactive}}=0$.
	Therefore, we prove the sufficiency that $ m=1$ or ${\mathcal{T}=0\Rightarrow \bar{\Delta} _{\rm Reactive}} = {\bar{\Delta} _{\rm Proactive}}$.
	\subsubsection{Proof of Necessity}
	Assume $m>1$ and $\mathcal{T}\ne0$, we can easily obtain from (\ref{HARQvs}) that ${\bar{\Delta} _{\rm Reactive}} - {\bar{\Delta} _{\rm Proactive}}>0$. This is equivalent to the necessity that ${\bar{\Delta} _{\rm Reactive}} = {\bar{\Delta} _{\rm Proactive}} \Rightarrow m=1$ or $ \mathcal{T}=0$.
	
	
	
	\section{Constructing the Mapping in Lemma \ref{mappingnp}}\label{appendixe}
	A constructing mapping from $\bf p$ to $\bf n$ is shown below:
	
	Step 1: Find the indexes of all the zero-value positions of vector ${\bf p} = \left(p_1,p_2,\cdots,p_{n_{\max}-n_{\min}+1}\right)$, and store them in an empty set $\mathcal{A}$. For example, if $p_i=0$, then $i$ is stored into $\mathcal{A}$.
	
	Step 2: Sort the elements in set $\mathcal{A}$ in the ascending order and denote the ordered elements as a vector $\left(a_1, a_2, \dots, a_{\left| {\mathcal{A}} \right|}\right)$, with $a_1 < a_2 < \cdots < a_{\left| {\mathcal{A}} \right|}$.
	
	Step 3: The vector $\bf n$ can be obtained by
	\begin{equation}
	n_i=a_i+n_{\min}-1, i=1,2,\cdots,\left|\mathcal{A}\right|.
	\end{equation}
	
	The above process is reversible, and the mapping from $\bf n$ to $\bf p$ is elaborated below:
	
	Step 1:  Construct the vector ${\bf a} = \left(a_1,a_2,\cdots,a_{m}\right)$ by
	\begin{equation}
	a_i=n_i-n_{\min}+1, i=1,2,\cdots,m.
	\end{equation}
	
	Step 2: Initialize ${\bf p} = {\bf 1}^{n_{\max}-n_{\min}+1}$ and let $p_{a_i}=0, i = 1, 2, \dots, m$, the vector $\bf p$ is obtained.
	
	\bibliographystyle{IEEEtran}
	\bibliography{reference}

\end{document}